\documentclass[12pt]{article}
\usepackage[margin=2.9cm]{geometry}

\usepackage[T1]{fontenc}
\usepackage[utf8]{inputenc}
\usepackage[english]{babel}
\usepackage{amsmath,amssymb, amsthm}
\usepackage[bookmarksopen,colorlinks,linkcolor=blue]{hyperref}
\usepackage{algorithm}
\usepackage{algpseudocode}
\usepackage{authblk}
\usepackage[shortlabels]{enumitem}
\usepackage{graphicx}

\usepackage{lmodern}
\usepackage{color}

\author[1]{Alexander Kozachinskiy\thanks{akozachinskiy@hse.ru
}}\author[1,2,3]{
Mikhail Vyalyi\thanks{vyalyi@gmail.com}}

\affil[1]{National Research University Higher School of Economics, Russian Federation}
\affil[2]{Moscow Institute of Physics and Technology, Russian Federation}
\affil[3]{Dorodnicyn Computing Centre, FRC CSC RAS, Russian Federation}
\title{An application of communication complexity, Kolmogorov complexity and extremal combinatorics to parity games}
\date{}

\newtheorem{theorem}{Theorem}
\newtheorem{corollary}[theorem]{Corollary}
\newtheorem{lemma}[theorem]{Lemma}
\newtheorem{proposition}[theorem]{Proposition}

\renewcommand{\le}{\leqslant}
\renewcommand{\ge}{\geqslant}

\let\sm\setminus
\def\X{\mathcal {X}}
\def\F{\mathcal {F}}
\def\G{\mathcal {G}}
\def\I{\mathcal {I}}
\def\P{\mathcal {P}}
\def\A{\mathcal {A}}
\def\B{\mathcal {B}}
\def\maj{\sqsubseteq}

\def\Pr{\mathop{\mathbf{Pr}}}

\begin{document}

\maketitle

\begin{abstract}
So-called separation automata are in the core of  several recently invented quasi-polynomial time algorithms for parity games. An explicit $q$-state separation automaton implies an algorithm for parity games with running time polynomial in $q$. It is open  whether a polynomial-state separation automaton exists. A positive answer will lead to a polynomial-time algorithm for parity games, while a negative answer will at least demonstrate impossibility 
to construct such an algorithm using separation approach. 

In this work we prove exponential lower bound 
for a restricted class of  separation automata.
Our technique combines communication complexity and Kolmogorov complexity.  One of our technical contributions belongs to extremal combinatorics. Namely, we prove  a new upper bound on the product of sizes of two families of sets with small pairwise intersection.
\end{abstract}

\section{Introduction}

Applications of Communication Complexity (CC) in Formal Language
Theory (FLT) are well-known. Apparently, the most important one is obtaining  lower
bounds on state complexity of non-deterministic automata (NFA) (see, e.g., the
monograph~\cite{hromkovic97}).
CC is also applied to analysis of
nondeterminism measures in finite automata~\cite{HSKKSch2002} and for
a number of other  problems in FLT (e.g., see~\cite{Ada2010} for
bounds on nondeterministic communication complexity of regular
languages).
Also, it is worth to mention that  lower bounds
on memory used in streaming algorithms, an~another important
application of CC (see
the book~\cite{RE2019}), 
can be viewed as lower bounds on  size of probabilistic automata of
specific form.
Note that most of these applications have equivalent combinatorial counterparts
(see discussion in~\cite{GruberHolzer06}).

In this paper we extend the applications of CC to separation problems
for \emph{safety automata}. These automata accept or reject infinite
words. They appeared in recent developments in algorithmic game
theory. More exactly, safety automata play an important role in 
analysis of quasi-polynomial algorithms for solving parity
games~\cite{czerwinski2019universal} (see details below). 

We are interested in state complexity of deterministic safety automata
separating a pair of languages (in the sequel, \emph{separation
automata}). There are no convenient tools for this task. We propose an
approach based on \emph{time restriction}. Our technique gives lower
bounds on state complexity of separation automata that accept a word
after reading a~sufficiently short prefix of an infinite word. 

These lower bounds are based on lower bounds for  multi-party
nondeterministic communication complexity in the number-in-hand model.
But, in contrast with the previous works, we do not give a direct way
to convert a small separation automaton to a protocol solving an
appropriate communication problem. Our approach uses also the ideas of
the fooling set technique. We conclude from lower bounds for a communication problem that a~small automata cannot separate a specific family of pairs of finite languages.  In the definition of this family we use Kolmogorov complexity to control the size of communication protocols. 
 Next step is to use this family to construct  a~pair of words  fooling  a~small separation automaton.
The family is used multiple times, and each time we have to manage Kolmogorov complexity by exploiting the fact that the automaton has few states.

We hope that the approach presented in this paper has a~potential to get
more strong bounds for separation automata as well as to be
applied for other problems in FLT.

To present our results in more details, we need a brief introduction to the area of
parity games and to separation approach in solving parity games.


\subsection{Parity games}

For a game with two competitive players one can consider a problem of deciding which player has a winning strategy.
Solving parity games is a~classical example when this problem lies in NP$\cap$coNP yet for which no polynomial-time algorithm is known. 
To specify an instance of a~parity game one needs  to specify:
\begin{itemize}
\item $n$-node directed graph   in which any node has at least one outgoing-edge;
\item indicated initial node;
\item labeling of edges by integers from $\{1, 2, \ldots, d\}$ (\emph{priorities});
\item partition of nodes into two parts, $V_0$ and $V_1$.
\end{itemize}

There are two players named \emph{Player 0} (Even) and \emph{Player 1} (Odd). 
A~position of a game 
is specified by a node of a graph.  It is possible to move from node $u$ to node $v$ if and only if $(u, v)$ is an edge of a graph. For each node $u$ it is predetermined which player makes a move in $u$. Namely, Player 0 makes a move in $V_0$ and Player 1 makes a move in $V_1$.
 
Since all nodes have out-going edges, 
a~play
can always last for infinite number of moves. In this way we obtain an infinite sequence of nodes $\{v_k\}_{k = 1}^\infty$ visited by players. We can also look at the sequence of corresponding priorities. Namely, let $l_k$ be a priority of an edge $(v_k, v_{k + 1})$. Winning conditions in parity game are the following: Player $i$ wins if and only if $$\limsup\limits_{k\to\infty} l_k \equiv i \pmod{2}.$$ 

Such a winning condition is Borel, which means due to Martin's theorem (\cite{martin1985purely}) that either Player $0$ or Player $1$ has a winning strategy. Moreover, it turns out  that a player having a~winning strategy in a parity game has also a \emph{memoryless} winning strategy, i.e. one in which  every move depends only on  a current node (\cite{emerson1991tree, mostowski1991games}). This fact means a lot for the complexity of  $\mathsf{ParityGames}$, a problem of determining the winner of a parity game.  Namely, due to this fact $\mathsf{ParityGame}$ is in NP$\cap$coNP (a short certificate for a player is 
his/her
memoryless winning strategy). More involved argument shows that actually $\mathsf{ParityGames}$ is in UP$\cap$coUP (\cite{jurdzi1998deciding}).

All this leaves a hope that $\mathsf{ParityGames}$ is solvable in polynomial time. 
Yet this is still an open problem. A~lot of work 
was done to improve an obvious $n^n$-time algorithm checking all memoryless strategies (see, e.g., \cite{petersson2001randomized, jurdzinski2008deterministic, mcnaughton1993infinite, schewe2007solving}).  This finally led in 2017 to a \emph{quasi-polynomial} time algorithm for $\mathsf{ParityGames}$:

\begin{theorem}[Calude et.\,al., \cite{calude2017deciding}]
$\mathsf{ParityGames}$ with $n$ nodes and $d$ priorities can be solved in $n^{O(\log d)}$ time.
\end{theorem} 

Although we 
made no assumptions on $d$, it is clear that we can always reduce a given instance of parity game to one in which $d$ is linear in the number of edges. Thus in a worst case the algorithm of Calude et.\,al. takes $n^{O(\log n)}$ time.

\subsection{Separation approach to parity games}

Since the paper of Calude et.\,al., several other quasi-polynomial time algorithms were invented for $\mathsf{ParityGames}$ (\cite{jurdzinski2017succinct, fearnley2017ordered, lehtinen2018modal}). The paper of Czerwi{\'n}ski et.\,al. (\cite{czerwinski2019universal}) argues that all these works follow so-called \emph{separation approach}. Let us briefly summarize this approach.

The main idea is to reduce $\mathsf{ParityGames}$ to  \emph{reachability games}. To specify a reachability game one needs to specify a graph and mark some of its nodes as winning. The goal of one player is to visit one of winning nodes at least once. Correspondingly, the goal of the other player is to avoid winning nodes. 

A standard analysis of complete information games works also for reachability games, which leads to a polynomial-time algorithm for the these games. 
In separation approach, a~parity game on a graph $G$ with $n$ nodes is reduced 
to a reachability game on a \emph{product} of~$G$ and the transition graph of some specific deterministic finite automaton $\mathcal{A}$. 

The input alphabet of $\mathcal{A}$ is a set $\{1, \ldots, n\}\times\{1, \ldots, d\}$ (pairs of the form $\langle$a node of $G$, a priority$\rangle$). We use this alphabet to encode infinite paths in $n$-node labeled graphs with $d$ priorities. Namely, assume that an infinite path starts with node $v_1$, then goes to $v_2$, then to $v_3$ and so on. Moreover, assume that the priority of the first edge in a path is $p_1$, the priority of the second edge in a path is $p_2$ etc. 
Then this path corresponds to 
the infinite  sequence $(v_1, p_1) (v_2, p_2) (v_3, p_3)\ldots$ over the alphabet $\{1, \ldots, n\}\times\{1, \ldots, d\}$. In what follows by saying that $\mathcal{A}$ does something on an infinite path we mean that $\A$ does something on the input sequence corresponding to this path. 


To make a reduction correct, we impose
the following requirement on $\mathcal{A}$. There should be a state $q_{accept}$ of $\mathcal{A}$ with the following properties:
\begin{itemize}
\item $\mathcal{A}$ reaches $q_{accept}$ on all paths produced by  memoryless strategies of Player~$0$ which are winning for some $n$-node graph with $d$ priorities.
\item $\mathcal{A}$ never reaches $q_{accept}$ on any path produced by a  memoryless strategy of Player $1$ which is winning  for some $n$-node graph with $d$ priorities.
\end{itemize}

Automata satisfying the above requirements are called  \emph{separation automata}. It follows immediately from definition  that  a~memoryless winning strategy  in a~parity game on $G$ yields a winning strategy in a reachability game on a product of $G$ and $\mathcal{A}$, where $\mathcal{A}$ is a separation automaton and  winning nodes  correspond to a state $q_{accept}$. 
Thus indeed to solve a parity game on $G$ it is enough to solve such a reachability game, and this takes time polynomial in the number of states of $\mathcal{A}$. 

It is possible to simplify a little bit a definition of separation automata.
A~graph is called  \emph{even} (\emph{odd}) if the maximum of priorities along a cycle is even (odd) for all cycles. 
 Take any winning positional strategy of Player $i$ in a parity game on $G$. Notice that if we remove  from~$G$ all edges contradicting this strategy, then we obtain, depending on~$i$, either odd or even graph.

 In \cite{czerwinski2019universal}, Czerwi{\'n}ski et.\,al. define 
the following two languages consisting of infinite words over $\{1, \ldots, n\}\times \{1, \ldots, d\}$.  Denote by $\mathrm{EvenCycles}_{n,d} \subseteq (\{1, \ldots, n\}\times\{1, \ldots, d\})^\mathbb{N}$ the set of all inputs sequences to $\mathcal{A}$ which correspond to some infinite path in an even graph with at most $n$ nodes and $d$ priorities. Define $\mathrm{OddCycles}_{n,d}$ similarly. Now from the observation above it follows that we can define separation automata equivalently as follows: $\mathcal{A}$ should reach $q_{accept}$ on sequences from $\mathrm{EvenCycles}_{n,d}$ and should avoid $q_{accept}$ on sequences from $\mathrm{OddCycles}_{n,d}$. 

As far as we know, before \cite{czerwinski2019universal}   a formalization of separation approach appears in a textbook of Boja\'nczyk and Czerwi{\'n}ski \cite{automata_toolbox}. However, instead of $\mathrm{EvenCycles}_{n,d}$ and $\mathrm{OddCycles}_{n,d}$, they used another two languages, $\mathrm{EvenLoops}_{n,d}$ and $\mathrm{OddLoops}_{n,d}$\footnote{Actually,  \cite{automata_toolbox} contains no name for these two languages and we use a terminology of \cite{czerwinski2019universal}.}. Namely,  $\mathrm{EvenLoops}_{n,d}$ ($\mathrm{OddLoops}_{n,d}$) consists of all infinite paths in which the maximum of priorities between any two visits of a same node is always even (odd).
It is clear that $\mathrm{EvenCycles}_{n,d}\subsetneq\mathrm{EvenLoops}_{n,d}$ and $\mathrm{OddCycles}_{n,d}\subsetneq\mathrm{OddLoops}_{n,d}$. Thus it is easier to construct separation automata in a~sense of  \cite{czerwinski2019universal} than in a~sense of \cite{automata_toolbox}. Correspondingly, it is easier to obtain lower bounds against the latter than against the former. We stress that in this paper we  follow the approach of \cite{czerwinski2019universal}, i.e., we use  $\mathrm{EvenCycles}_{n,d}$ and $\mathrm{OddCycles}_{n,d}$.


To describe the main lower bound of \cite{czerwinski2019universal} we shall introduce two more sets of infinite sequences from  $(\{1, \ldots, n\}\times\{1, \ldots, d\})^\mathbb{N}$. Namely, let $\mathrm{LimSupEven}_{n, d}$ be the set of all sequences $(v_1, p_1) (v_2, p_2) (v_3, p_3) \in (\{1, \ldots, n\}\times\{1, \ldots, d\})^\mathbb{N}$ satisfying $\limsup\limits_{i\to\infty} p_i \equiv 0 \pmod{2}$. Define $\mathrm{LimSupEven}_{n, d}$ similarly. 
Again, it is clear that
$$\mathrm{EvenCycles}_{n,d}\subsetneq\mathrm{LimSupEven}_{n, d}, \qquad \mathrm{OddCycles}_{n,d}\subsetneq\mathrm{LimSupOdd}_{n, d}.$$
First  Czerwi{\'n}ski et.\,al. demonstrate that actually all quasi-polynomial time algorithms for parity games listed above provide a quasi-polynomial-state automaton separating $\mathrm{LimSupEven}_{n, d}$ from $\mathrm{OddCycles}_{n,d}$ (in the same sense of separation as above --- an automaton should reach $q_{accept}$ on sequences from the first set an avoid $q_{accept}$ from sequences of the second set). It is more than required in  separation approach  --- however, no quasi-polynomial state automaton doing ``no more than required'' is known.

 On the other hand  Czerwi{\'n}ski et.\,al. show that any automaton separating $\mathrm{LimSupEven}_{n, d}$ from $\mathrm{OddCycles}_{n,d}$ has $n^{\Omega(\log d)}$ number of states. This exactly matches known constructions. To obtain such a lower bound,  they introduce a combinatorial object called ``universal trees'' and show that automata  separating $\mathrm{LimSupEven}_{n, d}$ from $\mathrm{OddCycles}_{n,d}$ should contain a~universal tree within the set of its states. Then they prove a quasi-polynomial lower bound on universal trees. 

It is not clear how to generalize this technique 
to separation  $\mathrm{EvenCycles}_{n,d}$ from $\mathrm{OddCycles}_{n,d}$ 
(for which no better lower bound that just $n$ is known). One of the obstacles is that the lower bound based on universal trees works also for non-deterministic automata.  At the same time separation of  $\mathrm{EvenCycles}_{n,d}$ and $\mathrm{OddCycles}_{n,d}$ is very easy with non-determinism allowed --- just guess a node appearing more then once and compute the maximum between two occurrences of this node.

\subsection{Our contribution}

We attack the question of obtaining lower bound on automata separating $\mathrm{EvenCycles}_{n,d}$ and $\mathrm{OddCycles}_{n,d}$. To do so we first relax a notion of separation automata by introducing an additional parameter $t$. 
Namely, recall that for any $w\in \mathrm{EvenCycles}_{n,d}$ a separation automaton should reach an accepting state on some finite prefix on $w$. The length of such prefix is not anyhow bounded. 
We suggest to simplify the problem and study it for automata in which such prefix is of length at most~$t$. 

More specifically, we say that a deterministic finite automaton separates $\mathrm{EvenCycles}_{n,d}$ from $\mathrm{OddCycles}_{n,d}$ in \textbf{time $\mathbf{t}$} if for all $w\in(\{1, \ldots, n\}\times\{1, \ldots, d\})^\mathbb{N}$:
\begin{itemize}
\item if $w\in\mathrm{EvenCycles}_{n,d}$, then an automaton reaches $q_{accept}$ while reading $w_1w_2\ldots w_t$ and  always stays in $q_{accept}$ after that;
\item if $w\in\mathrm{OddCycles}_{n,d}$, then an automaton never  reaches $q_{accept}$  on $w$.
\end{itemize}

A requirement that an automaton stays in $q_{accept}$ forever after reading $w_1w_2\ldots w_t$ is not essential because we can make  $q_{accept}$ an absorbing state.

It is easy to see that a deterministic automaton with $q$ states separating    $\mathrm{EvenCycles}_{n,d}$ from $\mathrm{OddCycles}_{n,d}$
necessarily does it 
in $qn$-time (for the sake of completeness we include the proof in Appendix \ref{sec_reduction}). Thus a  lower bound $q$ on the size of separation automata working in time $t$ implies $\min\{t/n, q\}$ lower bound on the size of  unrestricted separation automata.

Even super-linear lower bounds  for unrestricted separation automata are not known. To obtain such bounds with our approach   we first have to prove a good lower bound for super-quadratic $t$.
Unfortunately,   lower bounds we obtain in this paper are reasonable only for $t =  O(n^{5/4})$. 

\begin{theorem}
\label{main_theorem}
Any deterministic finite automaton separating $\mathrm{EvenCycles}_{n,2}$ from $\mathrm{OddCycles}_{n,2}$ in time $t$ has $\exp\left(\Omega(n^5/t^4)\right)$ number of states. 
\end{theorem}

Notice that this theorem is true even for $d = 2$. The fact that our argument uses only $2$ priorities means that essentially new ideas are needed to obtain similar bound for super-quadratic $t$. Indeed, there exists a simple $O(n)$-state deterministic automaton, separating $\mathrm{EvenCycles}_{n,2}$ from $\mathrm{OddCycles}_{n,2}$ in $O(n^2)$-time (namely, accept if and only if at least $n+1$ priorities which are equal to $2$ have been already seen).

\subsection{Auxiliary results}

For our proof we define the following communication problem which is a variation of Disjointness problem. Fix $n, k\in\mathbb{N}$ and $\gamma > 0$. There are $k$ parties. The $i^{th}$ party receives a set $X_i\subseteq\{1, 2, \ldots, n\}$ of size $\lfloor n/k\rfloor$. It is promised that either $X_1, X_2, \ldots, X_k$ are disjoint or $\forall i, i^\prime\in\{1, \ldots, k\}\,\, |X_i\triangle X_{i^\prime}| \le \gamma \cdot \lfloor n/k\rfloor$. The goal of parties  is to output $1$ in the first case and $0$ otherwise.  We denote this problem by $\mathrm{DISJ}^\prime_{k,\gamma}(n)$.

We show the following lower bound on $\mathrm{DISJ}^\prime_{k,\gamma}(n)$:

\begin{theorem} 
\label{communication_lower_bound}  For all large enough $n$ and for all $k\in\{2, \ldots, n - 1\}$ and   $\gamma \in (0, 1)$ satisfying $\frac{k}{\gamma} \le  \frac{\sqrt{n}}{100}$ the non-deterministic communication complexity of $\mathrm{DISJ}^\prime_{k,\gamma}(n)$ is at least $\frac{\gamma^2 n}{10^4 \cdot k} - 2 \log_2(n)$. 
\end{theorem}

A similar problem (without restrictions on sizes of input sets) in the two-party setting was considered in~\cite{gruska2013communication}. We postpone proof of Theorem \ref{communication_lower_bound} to Section~\ref{sec_com}.

To show Theorem \ref{communication_lower_bound} we prove the following result from extremal combinatorics which is interesting on its own:
\begin{theorem}\label{low-intersections}

For all $n, a, t\in\mathbb{N}$ satisfying $t < a < n$ the following holds.
If $\mathcal{F}\subseteq\binom{[n]}{a}$ and $\mathcal{G}\subseteq\binom{[n]}{a}$ are such that $|F\cap G| \le t$ for all $F\in \mathcal{F}$ and $G\in\mathcal{G}$ , then 
     $$|\mathcal{F}| \cdot |\mathcal{G}| \le 32  a (n - a)\cdot e^{-(a - t - 1)^2/(20a)}
     \binom{n}{a}^2. $$
\end{theorem}
We postpone  proof of  Theorem \ref{low-intersections} to Section \ref{sec_int}.
 For a special case when $t = \Omega(n)$ and $a - t = \Omega(n)$ this bound can be found  in a classical work of Frankl and R{\"o}dl \cite{frankl1987forbidden}. Moreover, their result only requires that $|F\cap G| \neq t + 1$ for all $F\in\mathcal{F}, G\in\mathcal{G}$. 
However, the paper~\cite{frankl1987forbidden} does not contain a complete proof of this bound and it is unclear how to restore details omitted. Also, 
it is quite hard to turn a proof of  Frankl and R{\"o}dl into an explicit bound for sublinear $a$ and $t$.

\section{Preliminaries}

We denote the set $\{1, 2, \ldots, n\}$ by $[n]$ and 
 the set $\{a, a + 1, \ldots, b\}$ by $[a, b]$. By $2^{[n]}$ we mean  the set of all subsets of $[n]$ and by $\binom{[n]}{k}$ we mean  the set of all $k$-element subsets of $[n]$. Notation $X\triangle Y$ is used for the symmetric difference of two sets $X$, $Y$.

\subsection{Separation automata}

Let $\Sigma$ be a finite alphabet. For $w\in\Sigma^*\cup \Sigma^\mathbb{N}$ by $|w|$ we denote the length of $w$. We assume that subscripts enumerating letters of $w$ start with $1$, i.e., we write $w = w_1 w_2 w_3\ldots$  


A deterministic finite automaton $\A$ over $\Sigma$ is specified by a finite set $Q$ of its states, an indicated initial state $q_{start}\in Q$ and a transition function $\delta_\A\colon Q \times \Sigma \to Q$. 
As usual, we  extend $\delta_\A$ 
to be a function of the form $\delta\colon Q \times \Sigma^* \to Q$ by setting 
$\delta_\A(q, w_1\ldots w_p)$ 
to be a state reached 
by the automaton 
from $q\in Q$ after reading $w_1\ldots w_p\in \Sigma^*$. 

For $A, B\subseteq \Sigma^\mathbb{N}$, $A\cap B = \varnothing$, we say that a deterministic finite automaton $\A$ \emph{separates} $A$ from $B$ if there exists a state $q_{accept}\in Q$  such that for all $w = w_1w_2\ldots\in\Sigma^\mathbb{N}$ the following holds:
\begin{itemize}
\item if $w\in A$, then there exists $i_0\in\mathbb{N}$ such that $\delta_\A(q_{start}, w_1\ldots w_i) = q_{accept}$ for all $i\ge i_0$;
\item if $w\in B$, then for all $i\in\mathbb{N}$ it holds that $\delta_\A(q_{start}, w_1\ldots w_i) \neq q_{accept}$.
\end{itemize} 
We say that an automaton separates $A$ from $B$ in time $t$ if, instead of the first condition, the stronger one holds:  if $w\in A$, then  $\delta(q_0, w_1\ldots w_i) = q_{accept}$ for all $i \ge t$.

A \emph{game graph} with $n$ nodes and $d$ priorities is a pair $G = \langle E, \pi\rangle$, where 
\begin{itemize}
\item $E$ is a subset of $\{1, \ldots, n\}^2$ satisfying the following condition: for all $u\in\{1, \ldots, n\}$ there is $v\in\{1, \ldots, n\}$ such that $(u, v) \in E$;
\item  $\pi$ is a function of the form $\pi\colon E \to\{1, \ldots, d\}$.
\end{itemize}
  I.e., we consider $G$ as a directed graph in which nodes are elements of $\{1, \ldots, n\}$ and edges are elements of $E$. Moreover,  edge $e$ has a label $\pi(e) \in\{1, \ldots, d\}$ on it. Edge labels are called priorities. A game graph should satisfy the following requirement: 
 for each node,  there exists at least one out-going edge. 
We stress that we  allow loops but do not allow parallel edges\footnote{Our main lower bound holds for graphs without loops as well and the proof is easily adaptable. To simplify an exposition, we present a weaker result.}.


A game graph $G = \langle E, \pi\rangle$ is called even (odd), if the maximum of $\pi$ on every cycle of $G$ is even (odd).  More formally, $G$ is called even (odd) if for all $k\ge 1$ and $v_1, \ldots, v_k\in\{1, \ldots, n\}$ satisfying:
$$(v_1, v_2), (v_2, v_3), \ldots, (v_{k - 1}, v_k), (v_k, v_1) \in E,$$
it holds that:
$$\max\{\pi((v_1, v_2)), \pi((v_2, v_3)), \ldots, \pi((v_{k - 1}, v_k)), \pi((v_k, v_1))\} \mbox{ is even (odd).}$$ 

Now let us define two sets (languages) consisting of infinite words over an alphabet $\Sigma = \{1, \ldots, n\}\times\{1, \ldots, d\}$. These two languages will be called $\mathrm{EvenCycles}_{n,d}$ and $\mathrm{OddCycles}_{n,d}$. Namely, an infinite sequence $(v_1, l_1) (v_2, l_2) (v_3, l_3)\ldots \in(\{1, \ldots, n\}\times\{1, \ldots, d\})^\mathbb{N}$ belongs to $\mathrm{EvenCycles}_{n,d}$ 
if 
there exists an even game graph $G = \langle E, \pi\rangle$ with at most $n$ nodes and $d$ priorities such that for all $i\ge 1$ it holds that $(v_i, v_{i + 1})\in E$ and $\pi((v_i, v_{i + 1})) = l_i$. I.e., we put $(v_1, l_1) (v_2, l_2) (v_3, l_3)\ldots$ into $\mathrm{EvenCycles}_{n,d}$ if and only if this sequence can be realized as an infinite path in some even game graph 
with at $n$ nodes and $d$ priorities.

If, instead of $G$ being even, we require that $G$ is odd, we obtain a definition of $\mathrm{OddCycles}_{n,d}$.

\subsection{Communication complexity}

For our main  lower bound we use non-deterministic communication complexity in the \emph{number-in-hand} model, but let us start with the deterministic case. In the number-in-hand model there are $k$ parties and their goal is to compute some (fixed in advance, possibly partial) function $f\colon\mathcal{X}_1\times\ldots\times\mathcal{X}_k \to \{0, 1\}$, where sets $\mathcal{X}_1, \ldots, \mathcal{X}_k$ are finite. The $i^{th}$ party receives an element $X_i$ of $\mathcal{X}_i$ on input. Parties have a shared blackboard on which they can write binary messages. 
Blackboard is seen by all parties. 
A~deterministic protocol 
specifies at each moment of time:
\begin{itemize}
\item whose turn is to write on the  blackboard (depending on what is already written there);
\item a message of the corresponding party (which depends not only on what is written on the blackboard but also on the player's input).
\end{itemize}
In the end of the communication, parties output 
a single bit which is assumed to be the value of $f$ on $(X_1, \ldots, X_k)$. This bit is a function of the history of communication, i.e. it can be computed by an external observer who can see only the blackboard but does not see inputs of players.  The communication complexity of a deterministic protocol $\pi$ (denoted below by $CC(\pi)$) is the maximal  possible (over all inputs) number of bits written on the blackboard in~$\pi$. 

Now let us switch to non-deterministic protocols. The most convenient definition for us is the following one. A non-deterministic protocol is a set $\mathcal{P}$ of deterministic protocols. A~run of a non-deterministic protocol has two phases. 
At first phase parties guess $\pi\in\mathcal{P}$. 
The guess is public 
so all the parties 
have the same~$\pi$.  Then the parties run $\pi$ on $(X_1, \ldots, X_k)$. By the communication complexity of $\mathcal{P}$ we mean the following expression:
$$CC(\mathcal{P}) = \lceil \log_2\left(\left|\mathcal{P}\right| \rceil \right) + \max\limits_{\pi\in\mathcal{P}} CC(\pi).$$
In particular, besides  communication in $\pi$, the number of bits needed to specify $\pi$ also counts. 
For brevity, we use a term ``$c$-bit protocol'' for a protocol with the communication complexity at most $c$.

 We say that $\mathcal{P}$ computes $f$ if for all $(X_1, \ldots, X_k) \in\mathcal{X}_1\times \ldots \times\mathcal{X}_k$ it holds that:
\begin{itemize}
\item if $f(X_1, \ldots, X_k) = 1$, then there is $\pi\in\mathcal{P}$ such that $\pi$ outputs $1$ on $(X_1, \ldots, X_k)$;
\item if $f(X_1, \ldots, X_k) = 0$, then for all $\pi\in\mathcal{P}$ it holds that $\pi$ outputs $0$ on $(X_1, \ldots, X_k)$. 
\end{itemize} 
Finally, by the non-deterministic communication complexity of $f$ we mean the minimal $c\in\mathbb{N}$ such that 
there exists a $c$-bit 
non-deterministic communication protocol computing $f$. 

More formal introduction to the number-in-hand model 
can be found, for instance, 
in \cite[Chapter 5]{jukna2012boolean}. 
For our lower bound we use only a~very basic technique of \emph{monochromatic boxes}. This technique is  a generalization of a standard two-party monochromatic rectangle technique.
A box is a set of the form $\mathcal{F}_1\times\ldots\times\mathcal{F}_k$ for some $\mathcal{F}_1\subseteq\mathcal{X}_1, \ldots, \mathcal{F}_k\subseteq\mathcal{X}_k$.  We exploit the following feature of protocols:   a $c$-bit non-deterministic protocol computing $f$  induces a \emph{cover} of $\{(X_1, \ldots, X_k) \in \mathcal{X}_1\times\ldots\times\mathcal{X}_k : f(X_1, \ldots, X_k) = 1\}$ 
by at most $2^c$ boxes such that each box in the cover does not  contain a tuple on which $f$ is defined and takes value $0$.

\subsection{Kolmogorov complexity}

Consider  two binary strings $x$ and $y$.
Informally speaking, the  conditional   Kolmogorov complexity of  $x$ given $y$  is the minimal length of a program producing $x$ from $y$ (length is measured in bits). To define it formally, consider any partial computable function $D\colon\{0, 1\}^*\times\{0, 1\}^*\to\{0, 1\}^*$. Let $C_D(x|y)$ denote $$\min\{|p| : p\in\{0, 1\}^* \mbox{ and } D(p, y) = x\}.$$
(here, as above, $|p|$ stands for the length of $p$). So $C_D(x|y)$ can be viewed as a compressed size of $x$ given $y$ with respect to ``decompressor'' $D$.  Kolmogorov -- Solomonoff theorem states that there exists an ``optimal'' decompressor; more precisely, there is a partial computable function $D_0\colon\{0, 1\}^*\times\{0, 1\}^*\to\{0, 1\}^*$ such that for any partial computable function $D\colon\{0, 1\}^*\times\{0, 1\}^*\to\{0, 1\}^*$    there exists $A > 0$ such that for  all $x, y\in\{0, 1\}^*$ we have $C_{D_0}(x|y) \le C_D(x|y) + A$.  We fix any such $D_0$ and let 
$C(x|y) = C_{D_0}(x|y)$ 
be the Kolmogorov complexity of $x$ given $y$. We also define the unconditional Kolmogorov complexity of $x$ as the Kolmogorov complexity of $x$ 
given the empty word.

Let us list some standard properties of Kolmogorov complexity which will be used in this paper. Proofs of them can be found, for instance, in \cite{shen2017kolmogorov}.

\begin{proposition}
\label{kolm_number}
For any $z\in\{0,1\}^*$ the number of $x\in\{0,1\}^*$ satisfying $C(x|z) \le a$ is less than $2^{a + 1}$. 
\end{proposition}

\begin{proposition}
\label{conservation}
For any computable function $f(\cdot, \cdot)$ and for all $x, y\in \{0,1\}^*$ the following holds:
$$ C(f(x, y)|y) \le C(x|y) + O(1)$$
(constant hidden in $O(\cdot)$ depends only on $f$ but not on $x$ and $y$).
\end{proposition}

\begin{proposition}
\label{simple_prefix_free}
For all $m\in\mathbb{N}$ and for all $x_1, \ldots, x_m, y\in\{0, 1\}^*$ the following holds:
$$C(x_1, x_2, \ldots, x_m|y) \le O(1) + \sum\limits_{i = 1}^m \left( 2C(x_i|x_1, \ldots, x_{i - 1}, y) + 2\right)$$
(constant hidden in $O(\cdot)$ is absolute)\footnote{There is a more tight relation between the left and the right hand side known as ``chain rule''. However, Proposition \ref{simple_prefix_free} is enough for our purposes.}.
\end{proposition}

Kolmogorov Complexity can be defined not only for binary strings but for other 
``finite objects'', 
like tuples of strings, finite sets, graphs etc. To do so we have to fix some encoding of these objects by binary strings. Different encodings lead to the same complexity up to $O(1)$ additive term.

\section{Proof of Theorem \ref{main_theorem}}

We actually prove a more specified version of Theorem \ref{main_theorem}. 
\begin{theorem}
\label{restated}
For all large enough $n$ the following holds. If $ 8n \le t \le \frac{n^{5/4}}{10^3}$, then any deterministic 
finite automaton separating $\mathrm{EvenCycles}_{n,2}$ from $\mathrm{OddCycles}_{n,2}$ in time  $t$ has  more than $2^{ \frac{n^5}{(10^3 \cdot t)^4 }}$ states.
\end{theorem}

 Theorem \ref{main_theorem}, however, has no restrictions on $t$, unlike Theorem \ref{restated}. Nevertheless,  it is easy to see that Theorem \ref{restated} implies Theorem \ref{main_theorem}. For $t > \frac{n^{5/4}}{10^3}$ the lower bound of Theorem \ref{main_theorem} is just constant, and the constant lower bound is obvious. Next, theorem \ref{main_theorem} for $n \le t < 8n$ follows from Theorem \ref{restated} for $t = 8n$ (with some constant loss in the exponent).  Finally, we observe that for $t < n$ there is no deterministic finite automaton separating  $\mathrm{EvenCycles}_{n,2}$ from $\mathrm{OddCycles}_{n,2}$ in time  $t$ at all. Indeed, a word $(1, 1) (2, 1) \ldots (n - 1, 1)$ is a prefix of a sequence from $\mathrm{EvenCycles}_{n,2}$ and also a prefix of a sequence from $\mathrm{OddCycles}_{n,2}$.

 Now we proceed to a proof of Theorem \ref{restated}. Assume for contradiction that for some $n$ and $8n \le t \le \frac{n^{5/4}}{10^3}$ there exists a deterministic finite automaton $\A$ with at most $Q$ states separating  $\mathrm{EvenCycles}_{n,2}$ from $\mathrm{OddCycles}_{n,2}$ in time  $t$. Here $Q$ is defined as follows 
\begin{equation}
\label{q_bound}
Q = 2^{\left\lceil \frac{n^5}{(10^{3} \cdot t)^4} \right\rceil}.
\end{equation}
  
To obtain a contradiction we construct two  words  on which $\mathcal{A}$ comes into the same state. One word is a prefix of a sequence from $\mathrm{EvenCycles}_{n,2}$. Moreover, its length is at least $t$.  The other word is a prefix of a sequence from $\mathrm{OddCycles}_{n,2}$. This  gives a  contradiction with the fact that $\A$ separates $\mathrm{EvenCycles}_{n,2}$ from $\mathrm{OddCycles}_{n,2}$ in time $t$.

To explain the construction
let us introduce some notation.  For a finite $X\subseteq \mathbb{N}$ denote:
$(X, 1) = (x_1, 1) (x_2, 1) \ldots (x_m, 1) \in (\mathbb{N} \times\{1\})^*,$
where $x_1, x_2, \ldots, x_m\in \mathbb{N}$ are such that $x_1 < x_2 < \ldots < x_m$ and $X = \{x_1, x_2, \ldots, x_m\}$. Next, for a word $w = (v_1, 1)\ldots (v_m, 1)\in(\mathbb{N}\times\{1\})^*$ denote $v(w) = \{v_1, \ldots, v_m\}$. I.e., $w \mapsto v(w)$ operation, loosely speaking, is the inverse to $X \mapsto (X, 1)$ operation.

 Set 
\begin{gather}
\label{parameters}
n^\prime = \lceil n/2\rceil,  \qquad k = 20 \cdot  \left\lfloor\frac{t}{n}\right\rfloor, \qquad
 \gamma = \frac{1}{k}, \qquad a = \lfloor n^\prime/k\rfloor,\\
\label{D_and_I}
\begin{aligned}
\mathcal{D} &= \left\{ (X_1, \ldots, X_k) \in \binom{[n^\prime]}{a}^k : X_1, X_2, \ldots, X_k \mbox{ are disjoint}\right\}, \\
\mathcal{I}&= \left\{ (Y_1, \ldots, Y_k) \in \binom{[n^\prime]}{a}^k :  \forall i, i^\prime \in \{1, \ldots, k\}\,\,  |Y_i\triangle Y_{i^\prime}| \le \gamma a\right\}.
\end{aligned}
\end{gather}
Note that $\mathrm{DISJ}^\prime_{k,\gamma}(n^\prime)$ is a problem to output $1$ on $\mathcal{D}$ and output $0$ on $\mathcal{I}$.

For a tuple $\overline{X} = (X_1, X_2, \ldots, X_k) \in \mathcal{D}$ 
let $<_{\overline{X}}$ be 
the linear order on $X_1 \cup X_2 \cup \ldots \cup X_k$ drawn on Figure~\ref{fig:order}. 
\begin{figure}[hh!]
\centering
\includegraphics{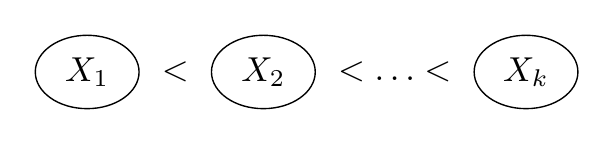}
\caption{$<_{\overline{X}}$ order}
\label{fig:order}
\end{figure}

Formally, we say that $p <_{\overline{X}} q$ 
if 
at least one of the following two conditions holds:
\begin{itemize}
\item $p\in X_i, q\in X_{i^\prime}$ for some $i,i^\prime\in[k], i < i^\prime$;
\item $p < q$ and $p, q\in X_i$ for some $i\in[k]$.
\end{itemize} 

Next, given $\overline{X} = (X_1, X_2, \ldots, X_k)\in\mathcal{D}$, let us say that a word $(v_1, l_1)\ldots (v_m, l_m) \in ([n]\times\{1, 2\})^*$ is $\overline{X}$-increasing if $v_1, v_2, \ldots, v_m \in X_1 \cup X_2 \cup \ldots \cup X_k$ and $v_1 <_{\overline{X}} v_2  <_{\overline{X}} \ldots  <_{\overline{X}} v_m$.

Finally, for $r\in\mathbb{N}$ let $\#_r$ denote a pair $(n^\prime + r, 2)$. We will 
use symbols $\#_r$ only for $r \le k/5 + 1$. 
It is easy to see from \eqref{parameters} and from the hypotheses of Theorem~\ref{restated} that
 $k = O(n^{1/4})$. This means that for any $r \le k/5 + 1$ it holds that $\#_r\in [n]\times \{1, 2\}$, i.e., $\#_r$ belongs to the input alphabet of $\A$. 

We are ready to formulate our main lemma.
\begin{lemma}
\label{long_inductive_lemma}
For some tuple $\overline{X} \in \mathcal{D}$ there are words $f^1, \ldots, f^{k/5}, g^1, \ldots, g^{k/5} \in([n^\prime]\times\{1\})^*$ satisfying 
the following conditions
\begin{itemize}
\item $v(f^1), \ldots, v(f^{k/5})$ are disjoint;
\item $g^1, \ldots, g^{k/5} \mbox{ are $\overline{X}$-increasing}$;
\item $|g^1| \ge 4n^\prime/7, \ldots, |g^{k/5}| \ge 4n^\prime/7$;
\item $\delta_{\mathcal{A}}(q_{start}, f^1\#_1 f^2 \#_2 \ldots f^{k/5} \#_{k/5}) 
= \delta_{\mathcal{A}}(q_{start}, g^1\#_1 g^2\#_2 \ldots g^{k/5}\#_{k/5})$.
\end{itemize}
Here $q_{start}$ is the initial state of $\A$. 
\end{lemma}
Let us explain how Lemma \ref{long_inductive_lemma} implies Theorem \ref{restated}.
 Take $\overline{X}\in\mathcal{D}$ and $f^1, \ldots, f^{k/5}, g^1, \ldots, g^{k/5} \in ([n^\prime]\times \{1\})^*$ satisfying Lemma \ref{long_inductive_lemma}. 
To obtain a contradiction it is enough to show that
\begin{align}
\label{odd_word}
f^1 \#_1 f^2 \#_2 \ldots f^{k/5} \#_{k/5} \mbox{ is a prefix of a word from } \mathrm{OddCycles}_{n,d},\\
\label{even_word}
g^1 \#_1 g^2 \#_2 \ldots g^{k/5} \#_{k/5} \mbox{ is a prefix of a word from } \mathrm{EvenCycles}_{n,d}.
\end{align}
Indeed, define
$$q^\prime = \delta_{\mathcal{A}}(q_{start}, f^1\#_1 f^2 \#_2 \ldots f^{k/5} \#_{k/5}) 
= \delta_{\mathcal{A}}(q_{start}, g^1\#_1 g^2\#_2 \ldots g^{k/5}\#_{k/5}).$$
By \eqref{odd_word} we have $q^\prime \neq q_{accept}$. On the other hand 
 the length of $g^1 \#_1 g^2 \#_2 \ldots g^{k/5} \#_{k/5}$ is at least $(k/5) \cdot (4n^\prime/7)$. By  \eqref{parameters} the last expression is at least   $4(\frac{t}{n} -1) \cdot \frac{2n}{7}$. In turn, from the formulation of Theorem~\ref{restated} we know that $t\ge 8n$. This implies that the length of $g^1 \#_1 g^2 \#_2 \ldots g^{k/5} \#_{k/5}$  is at least $t$. Due to \eqref{even_word} this means that $q^\prime = q_{accept}$, contradiction.

Let us at first show \eqref{odd_word}. Consider the following graph $G_{odd}$ (see Figure \ref{fig:odd_graph}). 
\begin{figure}[h!]
\centering
\includegraphics{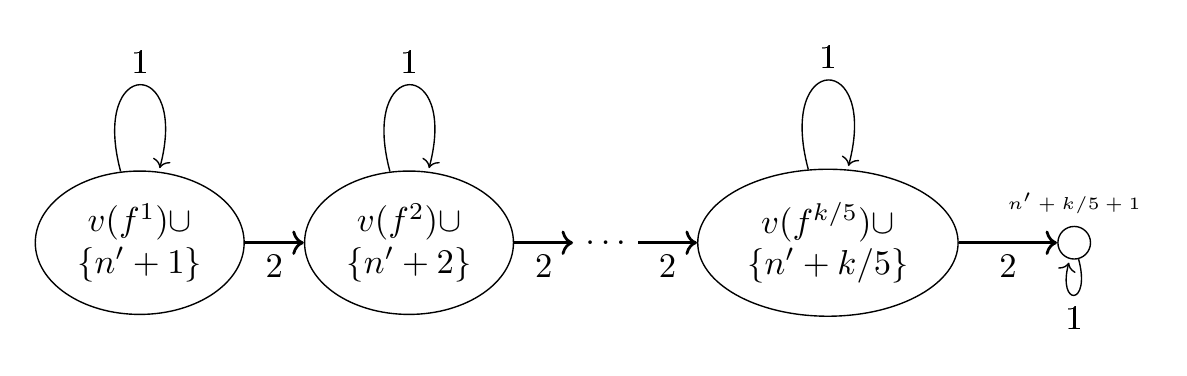}
\caption{A graph for  $f^1 \#_1 f^2 \#_2 \ldots f^{k/5} \#_{k/5}$. }
\label{fig:odd_graph}
\end{figure}

Nodes of this $G_{odd}$ are elements of $$v(f^1) \cup v(f^2) \cup \ldots \cup v(f^{k/5}) \cup \{n^\prime + 1, \ldots, n^\prime + k/5 + 1\}.$$
By Lemma \ref{long_inductive_lemma} sets $v(f^1), \ldots, v(f^{k/5})$ are disjoint subsets of $[n^\prime]$.  Let us specify edges of $G_{odd}$. First of all, for each $j\in[k/5]$ we draw all possible edges between nodes from $v(f^j) \cup \{n^\prime + j\}$ (including loops), each with priority~$1$.  Next, for all $j  < k/5$ we draw all edges  that start at a node from $v(f^j) \cup \{n^\prime + j\}$ and end at a node from $v(f^{j + 1}) \cup \{n^\prime + j + 1\}$,   each  with priority $2$. We also draw all edges that start at a node from $v(f^{k/5}) \cup \{n^\prime + k/5\}$ and end at $\{n^\prime + k/5 + 1\}$, each again with priority $2$. Finally, draw a loop
 at $n^\prime + k/5 + 1$ with priority~$1$ (we add this last loop to ensure that each node of $G_{odd}$ has at least one out-going edge).

It is easy to see from the construction that $G_{odd}$ is an odd game graph with at most $n$ nodes. Moreover, $f^1 \#_1 f^2 \#_2 \ldots f^{k/5} \#_{k/5}$ encodes a path in $G_{odd}$. Indeed, we move for some time in $v(f^1)$, then  through $n^\prime + 1$ we go to $v(f^2)$ and so on.  Thus \eqref{odd_word} is proved.

For  \eqref{even_word} it is extremely important that for some tuple $\overline{X}\in\mathcal{D}$ words  $g^1, \ldots, g^{k/5}$ are all $\overline{X}$-increasing. To see why, consider any even game graph $G$ with 2 priorities. If we remove all edges of priority $2$, we obtain an acyclic graph. Let $T$ be a topological ordering of the remaining graph. If we move in $G$ using only edges of priority $1$, then nodes we visit should increase in $T$. It is reflected in a fact that $g^1 \#_1 g^2 \#_2 \ldots g^{k/5} \#_{k/5}$ is split by $\#_1, \ldots, \#_{k/5}$ into $\overline{X}$-increasing words.

Now, to show \eqref{even_word} we define another graph, $G_{even}$ (see Figure \ref{fig:even_graph}).
\begin{figure}[h!]
\centering
\includegraphics{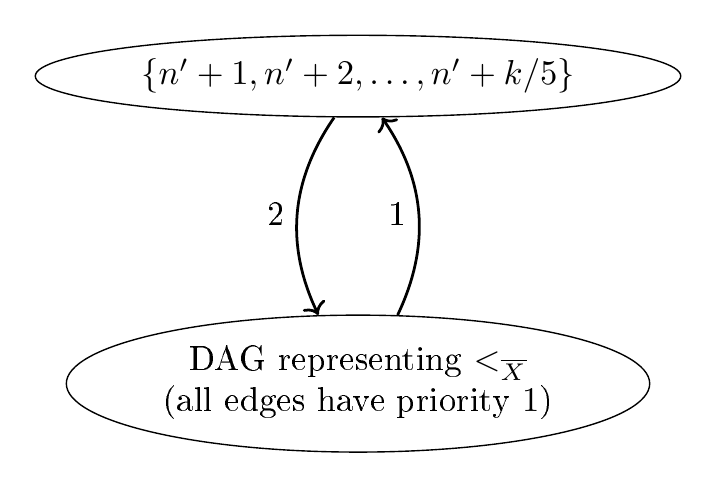}
\caption{A graph for $g^1 \#_1 g^2 \#_2 \ldots g^{k/5} \#_{k/5}$. }
\label{fig:even_graph}
\end{figure}
Its nodes are elements of $X_1\cup X_2 \cup \ldots \cup X_k \cup \{n^\prime + 1, \ldots, n^\prime + k/5\}$, where $\overline{X} = (X_1, X_2, \ldots, X_k)$. Next, let us specify edges of $G_{even}$. For all $u, v\in X_1 \cup X_2 \cup\ldots \cup X_k$ satisfying $u <_{\overline{X}} v$ we add an edge with priority $1$ from $u$ to $v$. Moreover, we draw all edges between $X_1\cup X_2 \cup \ldots \cup X_k$ and $\{n^\prime + 1, \ldots, n^\prime + k/5\}$ (in both directions). In particular, this ensures that each node of $G_{even}$ has at least one out-going edge. 
We assign  priority $1$ to the edges starting at $X_1\cup X_2 \cup \ldots \cup X_k$ and priority~$2$ to the edges starting in $\{n^\prime + 1, \ldots, n^\prime + k/5\}$.

Note that once we delete all edges with priority $2$ from $G_{even}$, we obtain an acyclic graph. Hence $G_{even}$ is an even game graph with at most $n$ nodes. On the other hand, since $g^1, g^2, \ldots, g^{k/5}$ are $\overline{X}$-increasing, it is easy to see that  $g^1 \#_1 g^2 \#_2 \ldots g^{k/5} \#_{k/5}$ corresponds to a path of $G_{even}$. Indeed, each $g^i$ represents a path at the bottom of the Figure \ref{fig:even_graph}. Once we reach the end of $g_i$, we go up with priority $1$. Then after reading $\#_j$, we go down.
  Thus \eqref{even_word} is proved. 

\subsection{Proof sketch of Lemma \ref{long_inductive_lemma}}
Here we give a proof sketch of Lemma \ref{long_inductive_lemma}. The proof is by induction. I.e., we first construct $f^1$ and $g^1$, then $f^2$ and $g^2$ and so on. 
A tuple 
 $\overline{X} = (X_1, \ldots, X_k)$, for which conditions of Lemma \ref{long_inductive_lemma} hold,
comes from the following 
\begin{proposition}
\label{overline_X}
 There exists $\overline{X} = (X_1, \ldots, X_k) \in\mathcal{D}$ such that for every state $q_0$ of $\A$ and for every $U \subseteq [n^\prime]$ satisfying
$C(U|n,t, \A) \le k\log_2(Q)$,
there exists $(Y_1, \ldots, Y_k) \in\mathcal{I}$ such that:
$$\delta_{\A}(q_0, (X_1\setminus U, 1) \ldots (X_k\setminus U, 1)) = \delta_\A (q_0, (Y_1\setminus U, 1)\ldots (Y_k\setminus U, 1)).$$ 
\end{proposition}
We derive this proposition from Theorem~\ref{communication_lower_bound} (a~lower bound for the problem $\mathrm{DISJ}^\prime_{k,\gamma}(n^\prime)$). 

Now, assume that $f^1, \ldots, f^{r - 1}, g^1, \ldots, g^{r - 1}$ satisfying Lemma \ref{long_inductive_lemma} are already constructed for some $r \le k/5$. Note that $f^1\#_1\ldots f^{r - 1}\#_{r -1}$ and $g^1\#_1\ldots g^{r-1} \#_{r-1}$ lead $\A$ into to the same state $q_0$. We shall construct $f^{r}, g^{r}\in([n^\prime]\times\{1\})^*$ satisfying the following conditions:
\begin{description}
\item[\textup{(a)}] $v(f^{r})$ is disjoint with $U = v(f^1) \cup v(f^2) \cup \ldots \cup v(f^{r-1})$;
\item[\textup{(b)}] $g^{r}$ is long enough (more precisely, its length should be at least $4n^\prime/7$) and $g^{r}$ is $\overline{X}$-increasing;
\item[\textup{(c)}] $\delta_\A(q_0, f^{r}) = \delta(q_0, g^{r})$.
\end{description}
To do so we apply Proposition \ref{overline_X} to $q_0$ and $U$ and set
$$f^{r} = (Y_1\setminus U, 1)\ldots (Y_k\setminus U, 1), \qquad g^{r} =  (X_1\setminus U, 1) \ldots (X_k\setminus U, 1),$$
where $(Y_1, \ldots, Y_k)\in\mathcal{I}$ is such that 
$$\delta_{\A}(q_0, (X_1\setminus U, 1) \ldots (X_k\setminus U, 1)) = \delta_\A (q_0, (Y_1\setminus U, 1)\ldots (Y_k\setminus U, 1)).$$ 
Now,  \textup{(a)}, \textup{(c)} and the second part of \textup{(b)}  immediately follow from the construction. Some explanation is needed only for the first part of \textup{(b)}. Recall that $(Y_1, \ldots, Y_k)\in\mathcal{I}$, which means that  $Y_1, Y_2, \ldots, Y_k$ are highly intersecting. This implies that $v(f^{r})$ is rather small, namely of size at most $2n^\prime/k$. I.e., each time we do an induction step, the size of $U$ increases by at most $2n^\prime/k$. 
Since the number of  increment steps is $k/5$, the size of $U$ is at most $2n^\prime/5$ at any moment. 
Now, recall that $g^{r} =  (X_1\setminus U, 1) \ldots (X_k\setminus U, 1)$ and $X_1, \ldots, X_k$ are disjoint $\lfloor n^\prime/k\rfloor$-elements subsets of $[n^\prime]$. This means that the length of $g^{r}$ is at least $k\cdot\lfloor n^\prime/k\rfloor - |U| \ge n^\prime - k - 2n^\prime/5 > 4n^\prime/7$.

The only remaining problem
is to show that Proposition \ref{overline_X} can indeed be applied to $U$. I.e., we have to ensure that the Kolmogorov complexity  of $U$ given $n, t$ and $\A$ is small.

Note that $U = v(f^1) \cup v(f^2) \cup\ldots \cup v(f^{r-1})$ is a function of $f^1, \ldots, f^{r-1}$. We will explain how to add a new $f^{r}$ in such a way that complexity of $U$ increases 
by approximately $\log_2(Q)$ bits. This guarantees that complexity of $U$ is at most   $\approx (k/5) \cdot \log_2(Q)$ at any moment.

So, we need a way to describe $f^{r}$ in just $\log_2(Q)$ bits
assuming that $f^1, \ldots, f^{r-1}$ (and also $n, t, \A$) are given. Recall how $f^{r}$ was constructed. Namely, note that $f^{r}$ is a function of $(Y_1 , \ldots, Y_k)$ and $U$. In turn, $U$ is a function of $f^1, \ldots, f^{r-1}$, so we only have a problem with $(Y_1, \ldots, Y_k)$. If we knew $\overline{X} = (X_1, \ldots, X_k)$, satisfying Proposition \ref{overline_X}, we could find $(Y_1, \ldots, Y_k)$ just by the brute-force search over $\mathcal{I}$. Indeed, first we compute $q_0 = \delta_\A(q_{start}, f^1\#_1 f^2\#_2\ldots f^{r-1}\#_{r - 1})$ (this yet does not require knowing $(X_1, \ldots, X_k)$). Then by emulating $\A$  we can find some $(Y_1, \ldots, Y_k)\in \I$ satisfying $\delta_{\A}(q_0, (X_1\setminus U, 1) \ldots (X_k\setminus U, 1)) = \delta_\A (q_0, (Y_1\setminus U, 1)\ldots (Y_k\setminus U, 1)).$

 However, it is unclear how to describe $\overline{X}$ in about $\log_2(Q)$ bits (even given $n,t$ and $\A$). One could argue that $\overline{X}$ can also be found by a brute-force search over $\mathcal{D}$. Nevertheless, this requires listing all $U$ of small Kolmogorov complexity. Unfortunately, Kolmogorov complexity is not computable.

The key observation here is that in the brute-force search algorithm for finding $(Y_1, \ldots, Y_k)$ described above we never used $(X_1, \ldots, X_k)$ as a whole. Instead, we only used
$q = \delta_{\A}(q_0, (X_1\setminus U, 1) \ldots (X_k\setminus U, 1))$ (for each $(Y_1, \ldots, Y_k)\in\I$ we check whether $q =  \delta_\A (q_0, (Y_1\setminus U, 1)\ldots (Y_k\setminus U, 1))$). Hence we can just give to the algorithm a $\log_2(Q)$-bit description of $q$. In this way we get a conditional $\log_2(Q)$-bit description of $(Y_1, \ldots, Y_k)$ given $f^1, \ldots, f^{r-1}$ and $n,t,\A$, as required.

\medskip

In the end of this subsection we provide  more details of the proof of Proposition~\ref{overline_X}. We define the following non-deterministic protocol involving $\A$.

\textbf{Description of the protocol $\mathbf{\mathcal{P}}$}. In this protocol there are $k$ parties and the $i${th} party receives a set $X_i\in\binom{[n^\prime]}{a}$.  At the beginning parties non-deterministically guess a state $q_0$ of $\mathcal{A}$ and a set $U\subseteq [n^\prime]$  satisfying $C(U|n,t,\mathcal{A}) \le k\log_2(Q)$.
 Then parties communicate in $k$ stages.   Stages are numbered from $1$ to $k$. At the $i$th stage the $i$th party writes  $\log_2(Q)$ bits specifying a state of $\mathcal{A}$ on the blackboard.  
Namely,
\begin{enumerate}
\item[] at the 
1st stage
the 1st party writes $q_1 = \delta_\A(q_0, (X_1\setminus U, 1))$;
\item[] at the 2nd stage the 2nd party writes $q_2 = \delta_\A(q_1, (X_2\setminus U, 1))$;
$$\vdots$$
\item[] at the $k$th stage the $k$th party writes $q_k = \delta_\A(q_{k - 1}, (X_k\setminus U, 1))$.
\end{enumerate}
Observe that
$$q_k = \delta_A(q_0, (X_1\setminus U, 1) (X_2\setminus U, 1)\ldots (X_k\setminus U, 1)).$$

After performing these $k$  stages 
parties finish communication. 
It remains to explain how the output of the protocol  $\mathcal{P}$ is computed. Parties output $1$ if and only if there is no $(Y_1, \ldots, Y_k)\in\mathcal{I}$ such that
$$q_k = \delta_A(q_0, (Y_1\setminus U, 1) (Y_2\setminus U, 1)\ldots (Y_k\setminus U, 1)).$$
In other words, parties output $1$ if and only if there is no input from $\mathcal{I}$ on which $\mathcal{P}$ produces the same $q_k$ for a guess $(q_0, U)$. 
\textbf{Description of the protocol is finished.}

It is easy to bound $CC(\mathcal{P})$. Parties communicate exactly $k\log_2(Q)$ bits. We should also add the number of bits needed to specify a non-deterministic guess of $\mathcal{P}$. For that we only need about $(k + 1)\log_2(Q)$ bits --- this is because the number of sets of complexity at most $k\log_2(Q)$ is smaller than $2^{k\log_2(Q) + 1}$. 
After that some tedious calculations show that with the choice of parameters as above $CC(\mathcal{P})$ is smaller than the non-deterministic communication complexity of $\mathrm{DISJ}^\prime_{k,\gamma}(n^\prime)$  (here we use the bound of Theorem \ref{communication_lower_bound}). This means that $\mathcal{P}$ does not compute  $\mathrm{DISJ}^\prime_{k,\gamma}(n^\prime)$. On the other hand, it is clear from the construction that $\mathcal{P}$ always outputs $0$ on any input from $\mathcal{I}$. Hence there should be a tuple $\overline{X} \in\mathcal{D}$ on which $\mathcal{P}$ 
 outputs $0$ for any possible non-deterministic guess. 
This is exactly what is needed from $\overline{X}$ in Proposition \ref{overline_X}.

We give a formal proof of Lemma \ref{long_inductive_lemma} in the next subsection.

\subsection{Proof of Lemma \ref{long_inductive_lemma}}
\label{proof_of_main_lemma}

To simplify the analysis below we need the following lower bound on separating $\mathrm{EvenCycles}_{n,2}$ from $\mathrm{OddCycles}_{n,2}$ without any time restrictions. 

\begin{proposition}
\label{simple_n_bound}
Any deterministic finite automaton separating $\mathrm{EvenCycles}_{n,2}$ from\linebreak $\mathrm{OddCycles}_{n,2}$ has at least $n + 1$ states.
\end{proposition}
\begin{proof}
Assume that a deterministic finite automaton $\B$ separates $\mathrm{EvenCycles}_{n,2}$  from $\mathrm{OddCycles}_{n,2}$. For $i = 0, 1, \ldots, n - 1$ define
$$q_i = \delta_\B(q_{start}, (1, 2) (2, 2)\ldots (i,2)),$$
where $q_{start}$ is the initial state of $\B$. Note that $(1, 2) (2, 2)\ldots (n - 1,2)$ is a prefix of a word from $\mathrm{OddCycles}_{n,2}$. Indeed, consider a graph which for $i\in [n - 1]$ has an edge from $i$ to $i + 1$ with priority $2$  and also has  loop with priority $1$ at node $n$. This means that $q_0 \neq q_{accept}$, $q_1 \neq q_{accept}$, $\ldots$, $q_{n - 1} \neq q_{accept}$. Now assume that $\B$ has at most $n$ states. Note that $q_0, q_1, \ldots, q_n$ are distinct from $q_{accept}$. It implies  that there are at most $n - 1$ possible values for each $q_0, q_1, \ldots, q_{n -1}$. Therefore 
there are $i, j\in \{0, 1, \ldots, n - 1\}$, $i < j$, such that 
$q_i = q_j$. Consider a graph $G$ with $n$ nodes which has all possible directed edges (including loops) and all of them have priority $2$.  Obviously, $G$ is an even game graph. Let $C_{i,j}$ be a cycle of $G$ obtained by going from $i + 1$ to $j$ and then back to $i + 1$ (in particular if $j = i + 1$, then $C_{i,j}$ is a loop at $j$). Consider an infinite path in $G$ which goes from $1$ to $i$ and then stays on $C_{i,j}$ forever. By definition, $\mathcal{B}$ should reach $q_{accept}$ on this path at some point. On the other hand, it is easy to see that the set of 
states visited by $\mathcal{B}$ 
on this path is $\{q_0, q_1, \ldots, q_i, \ldots, q_{j - 1}\}$. 
\end{proof} 

Recall that $\A$ separates $\mathrm{EvenCycles}_{n,2}$ from $\mathrm{OddCycles}_{n,2}$ in time $t$ and has at most $Q$ states. From 
Proposition \ref{simple_n_bound} we get
\begin{equation}
\label{q_is_not_too_small}
Q \ge n + 1
\end{equation}
(for the rest of the proof we only need  the fact that $Q$ is super-constant). From the hypotheses of Theorem~\ref{restated}  it is easy to derive the following bound: 
\begin{equation}
\label{k_asymptotic}
k = O(n^{1/4}).
\end{equation}

Now let us prove Proposition \ref{overline_X}.
\begin{proof}[Proof of Proposition  \ref{overline_X}]
Let $\mathcal{P}$ be a non-deterministic communication protocol defined on page 15. First let us establish that
$CC(\mathcal{P})$ is smaller than the non-deterministic communication complexity of $\mathrm{DISJ}^\prime_{k,\gamma}(n^\prime)$.

Let us start with the upper bound on the communication complexity of $\mathcal{P}$. By Proposition~\ref{kolm_number} 
there are at most $Q \cdot 2^{k\log_2(Q) + 1}$ possible non-deterministic guesses in $\mathcal{P}$.
 After making a~guess, parties communicate exactly $k\log_2(Q)$ bits. Therefore:
\begin{align*}
CC(\mathcal{P}) &\le \log_2(Q) + k\log_2(Q) + 1  + k\log_2(Q) = \left(2k + 1\right) \log_2(Q) + 1. 
\end{align*} 
The last expression is at most 
$3k \log_2(Q)$. 
Indeed, $k = 20 \lfloor t/n\rfloor$ by \eqref{parameters} and $t \ge  8n$ by 
hypotheses
of Theorem \ref{restated}. Hence $k\ge 160$ and $\frac{5k}{2} \ge 2k + 1$.  Note also that $\log_2(Q)$ 
is super-constant by \eqref{q_is_not_too_small}. Thus
\begin{align*}
 \left(2k + 1\right) \log_2(Q) + 1 &\le \frac{5 k}{2} \cdot \log_2(Q) + 1\\ 
&\le 3k \log_2(Q).
\end{align*}
In this way we conclude
\begin{equation}
\label{cc_p}
CC(\mathcal{P}) \le 3k \log_2(Q).
\end{equation}

Let us verify that $\frac{k}{\gamma} \le \frac{\sqrt{n^\prime}}{100}$. Indeed, again by \eqref{parameters} and 
by the hypotheses of Theorem \ref{restated} 
we have:
$$\frac{k}{\gamma} = k^2 \le \frac{400  \cdot t^2}{n^2} \le \frac{400 \cdot \frac{n^{5/2}}{10^6}}{n^2} = \frac{400 \cdot \sqrt{2}}{10^6} \cdot \sqrt{\frac{n}{2}}  \le \frac{400 \cdot \sqrt{2}}{10^6} \cdot \sqrt{n^\prime} < \frac{\sqrt{n^\prime}}{100}.$$
Hence by Theorem \ref{communication_lower_bound} the non-deterministic communication complexity of $\mathrm{DISJ}^\prime_{k,\gamma}(n^\prime)$ is at least
$$\frac{\gamma^2 n^\prime}{10^4 \cdot k} - 2\log_2( n^\prime) \ge \frac{\gamma^2 n}{2 \cdot 10^4 \cdot k} - 2\log_2(n) \ge  \frac{\gamma^2 n}{3 \cdot 10^4 \cdot k}.$$
In the first inequality we use the definition of $n^\prime$ (see \eqref{parameters}).  The second inequality holds because  $\gamma^2 n/k = n/k^3 = \Omega(n^{1/4})$ by \eqref{k_asymptotic}.

Thus by \eqref{cc_p} it remains to show that:
$$\log_2(Q) < \frac{\gamma^2 n}{9 \cdot 10^4 \cdot k^2} = \frac{n}{9\cdot 10^4 \cdot k^4}.$$
The right hand side by definition of $k$ (see \eqref{parameters}) is at least $$\frac{n}{9\cdot 10^4 \cdot \left(\frac{20 t}{n}\right)^4} \ge \frac{n^5}{10^{11}\cdot  t^4}.$$ In turn, the left hand side by definition of $Q$ (see \eqref{q_bound}) is at most
$$\log_2(Q) = \left\lceil \frac{n^5}{(10^3 \cdot t)^4} \right\rceil < \frac{n^5}{10^{12} \cdot t^4} + 1 \le 2 \cdot \frac{n^5}{10^{12} \cdot t^4} <  \frac{n^5}{10^{11} \cdot t^4},$$
where the second inequality holds because $t^4 \le \frac{n^5}{10^{12}}$ due to the hypotheses of Theorem~\ref{restated}.
Thus the fact that $CC(\mathcal{P})$ is smaller than the non-deterministic communication complexity of $\mathrm{DISJ}^\prime_{k,\gamma}(n^\prime)$ is proved.

This means that $\mathcal{P}$ does not compute $\mathrm{DISJ}^\prime_{k,\gamma}(n^\prime)$. In turn,
obviously $\mathcal{P}$  outputs~$0$ on any input from $\mathcal{I}$ for any possible guess. 
This means that there is $\overline{X} = (X_1, \ldots, X_k) \in\mathcal{D}$ such that $\mathcal{P}$ outputs~$0$ on the input $\overline{X}$ for any guess.  It is easy to see that this   is equivalent to the statement of Proposition \ref{overline_X}.
\end{proof}

To complete the proof of Lemma~\ref{long_inductive_lemma}, we introduce the algorithm $\mathbf{ALG_1}$.

Description of $\mathbf{ALG}_1$ involves a lot of notation which resembles the one used above, but with subscript $1$. This is to avoid confusion and to stress that $\mathbf{ALG_1}$ is independent of any other parameters. The latter is quite important due to our usage of Kolmogorov complexity.

 An  input to $\mathbf{ALG_1}$ consists of two parts:
\begin{itemize}
\item $n_1, t_1\in\mathbb{N}$, a deterministic finite automaton $\mathcal{A}_1$ with input alphabet $[n_1]\times\{1, 2\}$ and a tuple $\alpha = (f^1, \ldots, f^{j})$, where   $f^1, \ldots, f^{j} \in ([n^\prime_1] \times \{1\})^*$ and  $j \ge 0$ (when $j = 0$, we assume that $\alpha$ is empty);
\item a binary word $q \in \{0, 1\}^{\log_2(Q_1)}$.
\end{itemize}
Here
$$n^\prime_1 = \lceil n_1/2 \rceil, \qquad Q_1 = 2^{\left\lceil \frac{n_1^5}{(10^{3} \cdot t_1)^4} \right\rceil}$$
(i.e., $n^\prime_1$ and $Q_1$ are defined in the same way as  $n^\prime$ and $Q$ in \eqref{parameters} and \eqref{q_bound}).  The algorithm $\mathbf{ALG_1}$ also sets $k_1 = 20\left\lfloor\frac{t_1}{n_1}\right\rfloor$, $\gamma_1 = 1/k_1$, $a_1 = \lfloor n^\prime_1/k_1\rfloor$ and
$$\mathcal{I}_1= \left\{ (Y_1, \ldots, Y_{k_1}) \in \binom{[n^\prime_1]}{a_1}^{k_1} :  \forall i, i^\prime \in \{1, \ldots, k_1\}\,\,  |Y_i\triangle Y_{i^\prime}| \le \gamma_1 a_1\right\}.$$
(this is similar 
to the definitions of
$k, \gamma, a$ and $\mathcal{I}$ in \eqref{parameters} and \eqref{D_and_I}).

The algorithm $\mathbf{ALG_1}$ interprets  $q$ as a state of $\mathcal{A}_1$ (if there is more than $Q_1$ states in $\A_1$, then $\mathbf{ALG_1}$ halts and outputs ``not found''). The algorithm $\mathbf{ALG_1}$ computes
 $$U = v(f^1)\cup v(f^2) \cup\ldots \cup v(f^j), \qquad q_0 = \delta_{\mathcal{A}_1}(q_{start,1}, f^1\#_{1,1} \ldots f^j \#_{j,1}).$$
Here $q_{start,1}$ is the initial state of $\mathcal{A}_1$ and $\#_{1,1} = (n^\prime_1 + 1, 2), \ldots, \#_{j,1} = (n^\prime_1 + j, 2)$.
 Then  $\mathbf{ALG_1}$ tries to find $(Y_1, \ldots, Y_{k_1})\in\mathcal{I}_1$ satisfying the following condition:
$$q = \delta_{\mathcal{A}_1}(q_0, (Y_1\setminus U,1) (Y_2\setminus U, 1)\ldots (Y_{k_1}\setminus U, 1)).$$
Once any such $(Y_1, \ldots, Y_{k_1})$ is found, the algorithm $\mathbf{ALG_1}$ outputs a word $f = (Y_1\setminus U, 1)\ldots (Y_{k_1}\setminus U, 1)$. If there is no such $(Y_1, \ldots, Y_{k_1})$ at all, $\mathbf{ALG_1}$ halts and outputs ``not found''. \textbf{Description of the algorithm $\mathbf{ALG_1}$ is finished}.

For the rest of the proof, we assume that $\overline{X} = (X_1, X_2, \ldots, X_k)$ is a tuple satisfying the conditions of Proposition \ref{overline_X}. By  Proposition \ref{overline_X} and by the definition of  $\mathbf{ALG_1}$  we get:

\begin{proposition}
\label{alg_feature}
Take any  $f^1, \ldots, f^j \in ([n^\prime]\times\{1\})^*$. Define $U = v(f^1)\cup v(f^2) \cup \ldots \cup v(f^j)$ and 
$$q_0 = \delta_\A(q_{start}, f^1 \#_1 f^2 \#_2 \ldots f^j\#_j), \qquad q = \delta_\A(q_0, (X_1\setminus U, 1) (X_2\setminus U, 1) \ldots (X_k \setminus U, 1)).$$
Assume that $C(U|n,t,\A) \le k\log_2(Q)$. Then
$$\mathbf{ALG_1}((n, t, \A, (f^1, \ldots, f^j)), q) = (Y_1\setminus U, 1) (Y_2\setminus U,1)\ldots (Y_k\setminus U, 1)$$
for some $(Y_1, \ldots, Y_k)\in\mathcal{I}$ satisfying $q = \delta_\A(q_0, (Y_1\setminus U, 1) (Y_2\setminus U, 1) \ldots (Y_k \setminus U, 1))$.
\end{proposition}

To show Lemma \ref{long_inductive_lemma} it is enough to
 show that for every $r = 1, \ldots, k/5$ there are words $f^1, \ldots, f^r, g^1, \ldots, g^r \in([n^\prime]\times\{1\})^*$ satisfying 
the following conditions:
\begin{align}
\label{f_disjoint}
v(f^1), \ldots, v(f^r) \mbox{ are disjoint and } |v(f^1)| \le 2n^\prime/k, \ldots, |v(f^r)| \le 2n^\prime/k,\\
\label{f_complexity}
C(f^j|f^1, \ldots, f^{j - 1}, n, t, \A) \le  2  \log_2(Q) \mbox{ for } j = 1, \ldots, r,\\
\label{g_disjoint}
 |g^1|  \ge 4n^\prime/7, \ldots, |g^r| \ge 4n^\prime/7,\\
\label{g_order}
g^1, \ldots, g^r \mbox{ are $\overline{X}$-increasing},\\
\label{same_state}
\delta_{\mathcal{A}}(q_{start}, f^1\#_1 f^2 \#_2 \ldots f^r \#_r) 
= \delta_{\mathcal{A}}(q_{start}, g^1\#_1 g^2\#_2 \ldots g^r\#_r).
\end{align}
The proof is by induction on $r$. Induction base and induction step 
will be proved by the same argument. 
Namely, assume that $f^1, \ldots, f^{r - 1}, g^1, \ldots, g^{r - 1}$ satisfying (\ref{f_disjoint}--\ref{same_state}) are already constructed for some $r \le k/5$ (case $r = 1$ corresponds to the induction base). Define
$$U = v(f^1) \cup v(f^2) \cup \ldots v(f^{r - 1}),$$
$$q_0 = \delta_{\mathcal{A}}(q_{start}, f^1\#_1 f^2 \#_2 \ldots f^{r - 1} \#_{r - 1})$$
(for $r = 1$ we have $U = \varnothing$ and $q_0 = q_{start}$).
Note that by \eqref{same_state} we also have
$$q_0 =  \delta_{\mathcal{A}}(q_{start}, g^1\#_1 g^2 \#_2 \ldots g^{r - 1} \#_{r - 1}).$$

It is enough to  construct $f^r, g^r \in ([n^\prime]\times\{1\})^*$ satisfying:
\begin{align}
\label{f_disjoint_2}
v(f^r) \cap U &= \varnothing \mbox{ and } |v(f^r)| \le 2n^\prime/k, \\
\label{f_complexity_2}
C(f^r| f^1, \ldots, f^{r - 1}, n, t, \mathcal{A}) &\le  2\log_2(Q),\\
\label{g_disjoint_2}
 |g^r| &\ge 4n^\prime/7,\\
\label{g_order_2}
g^r &\mbox{ is $\overline{X}$-increasing},\\
\label{same_state_2}
\delta_\A(q_0, f^r) &= \delta_\A(q_0,g^r).
\end{align}
We define $g^r$ as follows:
$$g^r = (X_1\setminus U, 1) (X_2 \setminus U, 1) \ldots (X_k\setminus U, 1). $$
At first, we derive \eqref{g_disjoint_2} and \eqref{g_order_2}. The latter is clear from construction. As for the former, recall that $(X_1, \ldots, X_k) \in\mathcal{D}$, i.e., $X_1, \ldots, X_k$ are disjoint. Hence $|g^r| = |(X_1\cup X_2 \cup \ldots \cup X_k) \setminus U|$. The last expression is at least $k \cdot \lfloor n^\prime/k\rfloor - |U|$,
By \eqref{f_disjoint} and by definition of $U$ its size is at most $(k/5) \cdot (2n^\prime/k) = 2n^\prime/5$. As $k = O(n^{1/4})$ by \eqref{k_asymptotic}, we obtain $|g^r| \ge 4n^\prime/7$.  

It remains to derive \eqref{f_disjoint_2}, \eqref{f_complexity_2} and \eqref{same_state_2} (these conditions involve $f^r$ which is not yet defined). For that we first have to establish that
$C(U|n,t,\A) \le k\log_2(Q)$. By applying Proposition \ref{conservation} to a mapping, which takes a tuple of strings from $(\mathbb{N}\times\{1\})^*$, applies $v$ to them and takes the union, we get:
$$C(U|n,t,\A) \le C(f^1, f^2, \ldots, f^{r - 1}|n,t,\A) + O(1).$$
 By Proposition \ref{simple_prefix_free}, the right hand side of the last inequality is 
upperbounded by
$$O(1) + \sum\limits_{j = 1}^{r - 1} (2 C(f^j| f^1, \ldots, f^{j - 1}, n, t,\A) + 2).$$
The last sum by \eqref{f_complexity} is at most $(k/5) \cdot (4\log_2(Q) + 2) + O(1) \le k\log_2(Q)$. The last inequality  
holds because $k \ge 160$ (see the proof of Proposition \ref{overline_X}) and $Q$ is super-constant by \eqref{q_is_not_too_small}. 

Set $q = \delta_\A(q_0, g^r) = \delta_\A(q_0, (X_1\setminus U, 1) (X_2 \setminus U, 1) \ldots (X_k\setminus U, 1))$ and define
$$f^r = \mathbf{ALG_1}((n, t, \A, (f^1, f^2, \ldots,  f^{r - 1})), q).$$
Since we have proved that $C(U|n,t,\A)\le k\log_2(Q)$, from Proposition \ref{alg_feature} we  obtain that:
$$f^r = (Y_1\setminus U, 1) (Y_2\setminus U, 1) \ldots (Y_k\setminus U, 1)$$
for some $(Y_1, \ldots, Y_k)\in\mathcal{I}$ satisfying $q = \delta_\A(q_0, (Y_1\setminus U, 1) (Y_2\setminus U, 1) \ldots (Y_k \setminus U, 1))$. 

From that we immediately  get \eqref{same_state_2}. Indeed, $q =\delta_\A(q_0, g^r)$ by definition. On the other hand, $\delta_\A(q_0, f^r) =  \delta_\A(q_0, (Y_1\setminus U, 1) (Y_2\setminus U, 1) \ldots (Y_k \setminus U, 1)) = q$.

The first part of \eqref{f_disjoint_2} is once again obvious because  $f^r = (Y_1\setminus U,1) (Y_2\setminus U, 1)\ldots (Y_k\setminus U, 1)$. To show the second part of \eqref{f_disjoint_2} observe that 
$v(f^r) \subseteq Y_1 \cup Y_2 \cup \ldots \cup Y_k$.  
Hence
\begin{align*}
|v(f^r)| &\le |Y_1| + |Y_2\setminus Y_1| + \ldots + |Y_k\setminus Y_1| \\
&\le \frac{n^\prime}{k} + (k - 1) \frac{\gamma n^\prime}{k}  \le \frac{2n^\prime}{k}.
\end{align*}
Here in the second inequality we use the fact that $(Y_1, Y_2, \ldots, Y_k) \in \I$ and in the third inequality we use the definition of $\gamma$ (see \eqref{parameters}).

Finally, to  show \eqref{f_complexity_2} recall once again that
$$f^r = \mathbf{ALG_1}((n, t, \A, (f^1, f^2, \ldots,  f^{r - 1})), q).$$
Hence by the definition of conditional Kolmogorov complexity we have:
$$C(f^r| f^1, \ldots f^{r - 1}, n,t,\A) \le |q| + O(1) = \log_2(Q) + O(1) \le 2\log_2(Q),$$
where the last inequality  is due to \eqref{q_is_not_too_small}.

\section{Proof of Theorem \ref{low-intersections}}
\label{sec_int}

Let us  sketch  our proof of Theorem \ref{low-intersections}. First of all, for the sake of brevity we say that 
two families $\F, \G\subseteq\binom{[n]}{a}$ are  \emph{$t$-far} if $|F\cap G| \leq t$ for
all $F\in\mathcal{F}, G\in\mathcal{G}$ (so that any member of $\F$ is of Hamming distance at least $2a -  2t$ from any member of $\G$). 

\textbf{Step 1}. We  use a classical shifting technique of \cite{erdos1961intersection} to define so-called \emph{left-compressed} families.  We show that it is enough to demonstrate Theorem \ref{low-intersections} for the case when $\F$ is left-compressed  (Lemma \ref{reduction_to_compressed_families}).

\textbf{Step 2}. We observe  (Proposition \ref{ideals-filters})  that left-compressed families are ideals of a special partial order $\maj_a$ (see \cite{Bashov2011}) on a set $\binom{[n]}{a}$. 

\textbf{Step 3}. We give a necessary and sufficient condition for a family $\G \subseteq\binom{[n]}{a}$ to be $t$-far from  an ideal $\F$ of $\maj_a$  (Lemma \ref{FI-noncomparable}).

\textbf{Step 4}. Using this condition we give an upper bound on the probability 
that $\mathbf{X} \in\F$ and $\mathbf{Y}\in\G$ for two suitably chosen independent random variables $\mathbf{X}$ and $\mathbf{Y}$ (Lemma \ref{prob_bound}).  From that we easily deduce an upper bound on $|\F|\cdot |\G|$.

\subsection{Shifting and compression}

For every $i,j\in[n]$ we define so-called \emph{shifting operations} $s_{ij}$ and $S_{ij}$. Namely, $s_{ij}$  is a unary operation on the set of all subsets of $[n]$. Given $X\subseteq[n]$, the value of $s_{ij}(X)$ is defined as follows:
\[
s_{ij}(X) = \left\{
\begin{aligned}
  &(X\sm\{j\})\cup\{i\}, &&\text{if}\ j\in X,\ i\notin X,\\
  &X, &&\text{otherwise.}
\end{aligned}\right.
\]
In turn, $S_{ij}$ is a unary operation on the set of all \emph{families} of subsets of $[n]$. Given $\mathcal{X}\subseteq 2^{[n]}$, we define the value of $S_{ij}(\mathcal{X})$ as follows:
\[
S_{ij}(\X) =\{ s_{ij}(X) : X\in\X, s_{ij}(X)\notin \X\}
\cup \{ X : X\in\X, s_{ij}(X)\in \X\}.
\]

Note that $s_{ij}$ preserves the size of a set, i.e., $|X| =|s_{ij}(X)|$ for all $X\subseteq[n]$.
Hence if a family $\mathcal{X}$ consists only of $a$-element subsets of $[n]$, then the same holds for $S_{ij}(\mathcal{X})$. 
  It is also easy to see   that $S_{ij}$ preserves the size of a family, i.e.,  $|\X| = |S_{ij}(\X)|$ for all $\X\subseteq 2^{[n]}$.

\begin{proposition}[Lemma 2.1 from \cite{Borg2014TheMP}]\label{shifts}
  Assume that  $1\leq i<j\leq n$ and  $\F, \G\subseteq\binom{[n]}{m}$ are $t$-far. Then  $S_{ij}(\F)$, $S_{ji}(\G)$ are also $t$-far. 
\end{proposition}

A family $\F\subseteq 2^{[n]}$ is said to be \emph{left-compressed} if $S_{ij}(\F) =
\F$ for all $i<j$.

\begin{lemma}
\label{reduction_to_compressed_families}
If $\F, \G\subseteq\binom{[n]}{m}$ are $t$-far, then there are $\F^\prime, \G^\prime\subseteq\binom{[n]}{a}$ satisfying the following three conditions: 
\begin{itemize}
\item $\F^\prime$ and $\G^\prime$ are $t$-far;
\item $|\F^\prime| = |\F|$ and $|\G^\prime| = |\G|$;
\item $\F^\prime$ is left-compressed.
\end{itemize}
\end{lemma}

It is easy to deduce the last lemma from Proposition \ref{shifts}. Indeed, apply $S_{ij}$ to $\F$ and $S_{ji}$ to $\G$ until $S_{ij}(\F) \neq \F$ for some $i < j$. To show that this can be done only finite number of times observe that
$$\sum\limits_{A\in S_{ij}(\F)} \sum\limits_{i \in A} i < \sum\limits_{A\in\F} \sum\limits_{i \in A} i,$$
whenever $S_{ij}(\F) \neq \F$. 
 The proof  can also be found in \cite{Borg2014TheMP} (see the last two paragraphs before Section 3).
%
%

\subsection{Auxiliary order}

For $X\subseteq [n]$ and $1\le i \le |X|$ define $m(X, i)$ to be the $i$th smallest element of $X$. Also define $m(X, 0) = 0$. 

 For any  $l\in[n]$ 
we define the partial order $\maj_l$  on the set $\binom{[n]}{l}$ as follows (see~\cite{Bashov2011}):  $X\maj_l Y$ if $m(X, i) \le m(Y, i)$ for all $1 \le i \le l$.

\begin{proposition}
\label{move2left}
  Let $X =\{x_1,\dots, x_l\}\in\binom{[n]}{l}$  and $Y\in\binom{[n]}{l}$ be such that
  $x_i\leq m(Y,i)$ for all $1\leq i\leq l$. Then $X\maj_l Y$.
\end{proposition}

Note that $x_i$ in this proposition  are not ordered. In other words, a
smaller set w.r.t. this order can be produced by decreasing values of
some elements of a~set.

\begin{proof}[Proof of Proposition \ref{move2left}]
Take any $i\in [l]$. Let $j$ be the largest element of $\{0, 1, \ldots, l\}$ satisfying $m(X, j) \le m(Y, i)$. Note that $j$ is equal to the size of $X\cap [1, m(Y, i)]$. On the other hand, we have $x_1 \le m(Y, 1), \ldots, x_i\le m(Y, i)$. Hence $x_1, \ldots, x_i \in  X\cap [1, m(Y, i)]$, which means that $j = |X\cap [1, m(Y, i)]| \ge i$. Therefore $m(X, i) \le m(X, j) \le m(Y, i)$.
\end{proof}

\begin{proposition}
\label{move2right}
 Let $X\in\binom{[n]}{l}$  and $Y=\{y_1, \ldots, y_l\}\in\binom{[n]}{l}$ be such that
  $m(X, i)\leq y_i$ for all $1\leq i\leq l$. Then $X\maj_l Y$.
\end{proposition}
\begin{proof}
Apply Proposition \ref{move2left} to $X^\prime = \{n - y_l + 1, n - y_{l - 1} + 1, \ldots, n - y_1 + 1\}$ and $Y^\prime = \{n - j + 1 : j\in X\}$. 
\end{proof}

Recall that an \emph{ideal} $\A$ of a partially ordered set $\P$ is a
downward-closed subset of $\P$: if $x\leq_P y$ and $y\in \A$, then
$x\in\A$. 

\begin{proposition}[Proposition 3 in \cite{Bashov2011}]\label{ideals-filters}
  A left-compressed family $\F\subseteq \binom{[n]}{a}$ is an ideal of the order $\maj_a$. 
\end{proposition}

For reader's convenience we also give here a proof sketch of Proposition \ref{ideals-filters}. If $\F$ is not an ideal of $\maj_a$, then for some $B\in \F$ there is $A\in\binom{[n]}{a}\setminus \F$ immediately preceding $B$ with respect to $\maj_a$. It is not hard to see that $A$ can be obtained from $B$ after decreasing some element of $B$ (say, $i$) by one. Then $s_{i-1, i}(B) = A$ and hence $\F$ is not left-compressed.

 So, it suffice to prove
Theorem~\ref{low-intersections} for a pair $(\F,\G)$ in which  $\F$ is an ideal  of the order $\maj_a$.

\subsection{Characterizing families which are  $t$-far from ideals}

Define the \emph{$j$-left border} $L_j(X)$ and the \emph{$j$-right border}
$R_j(X)$ of a set $X\subseteq\binom{[n]}{a}$
as 
\[
L_j(X) = \{m(X,i): 1\leq i\leq j\}; \quad 
R_j(X) = \{m(X,i): a-j+1\leq i\leq a\}.
\]
In other words, $L_j(X)$ consists of $j$ smallest elements of $X$ and
$R_j(X)$ consists of $j$ largest elements of~$X$.

\begin{lemma}\label{FI-noncomparable}
Let  $\F\subseteq \binom{[n]}{a}$ be an ideal of $\maj_a$. 
Then for any $\G\subseteq\binom{[n]}{a}$ the following two conditions are equivalent:
\begin{description}
\item[\textup{(a)}] $\F$ and $\G$ are $t$-far;
\item[\textup{(b)}]  $L_{t + 1}(G) \not\maj_{t + 1} R_{t + 1}(F)$ for all $F\in\F$ and $G\in\G$.
\end{description}
\end{lemma}
\begin{proof}

\textbf{(b)} $\mathbf{\implies}$ \textbf{(a)}. Assume for contradiction that $\mathcal{F}$ and $\G$ are \emph{not} $t$-far. Hence there are $F\in\F$ and $G\in\G$ such that $|F\cap G| \ge t + 1$. Let $X$ be any $(t + 1)$-element subset of $F\cap G$. Then obviously we have that  $L_{t + 1}(G) \maj_{t + 1} X \maj_{t + 1} R_{t+1}(F)$, contradiction. 

\textbf{(a)} $\mathbf{\implies}$ \textbf{(b)}.  Assume for contradiction that there are $F\in\F$ and $G\in \G$ such that $L_{t + 1}(G)\maj_{t + 1} R_{t + 1}(F)$.  Define
$$\F^\prime = \{ F^\prime \in \F : L_{t + 1}(G)\maj_{t + 1} R_{t + 1}(F^\prime)\}.$$ 
By definition $F\in\F^\prime$, i.e., $\F^\prime$ is non-empty. Let $F_0$ be any minimal element of $\F^\prime$ with respect to $\maj_a$, i.e., assume there is no $F^\prime\in \F^\prime$, $F^\prime\neq F_0$ such that $F^\prime \maj_a F_0$. To obtain a~contradiction it is enough to show that $|F_0 \cap G| \ge t + 1$ (this would mean that $\F$ and $\G$ are not $t$-far).

Assume that $|F_0 \cap G| < t + 1$. Hence there is an element of $L_{t + 1}(G)$ which is not in $F_0$. Namely, there is $i\in\{1, 2, \ldots, t + 1\}$ such that $m(G, i) \notin F_0$. Define
$$F_1 = (F_0\setminus \{m(F_0, a - t - 1 + i)\}) \cup\{ m(G, i)\}.$$
First of all, observe that $|F_1| = |F_0| = a$ (this is because $m(F_0, a - t - 1 + i) \in F_0$ and $m(G, i)\notin F_0$). 
Let us check that the following three claims hold:
\begin{align}
\label{in}
F_1 &\in \F^\prime \\
\label{neq}
F_1 &\neq F_0 \\
\label{maj}
F_1 &\maj_a F_0.
\end{align}
These three claims give a contradiction with minimality of $F_0$. 

The simplest one is \eqref{neq} --- observe that $F_1$ contains $m(G, i)$ and $F_0$ does not. 

Now, let us show \eqref{maj}. Recall that $F_0 \in \F^\prime$, i.e., $L_{t + 1}(G)\maj_{t + 1} R_{t + 1}(F_0)$. Hence $m(G, i) = m(L_{t + 1}(G), i)\le m(R_{t + 1}(F_0), i) =  m(F_0, a - t - 1 + i)$, i.e., $F_1$ is obtained from $F_0$ by removing a bigger element and adding a smaller element (which originally was not in $F_0$). Hence by Proposition \ref{move2left} we have that $F_1\maj_a F_0$.

To show \eqref{in} let us at first show that $F_1\in\F$. Indeed, $\F$ is an ideal of $\maj_a$ and 
$F_0\in\F^\prime\subseteq \F$. 
Hence by \eqref{maj} we have that $F_1\in \F$. To show that actually $F_1\in \F^\prime$ 
we have to prove that $L_{t + 1}(G)\maj_{t + 1} R_{t + 1}(F_1)$. 
Define
$$X = (R_{t + 1}(F_0)\setminus \{m(F_0, a - t - 1 + i)\}) \cup\{ m(G, i)\}.$$
Observe that $X$ is a $(t + 1)$-element subset of $F_1$.  Note that $m(G, i)  = m(L_{t + 1}(G), i) \notin R_{t + 1}(F_0)$ and $m(F_0, a - t - 1 + i) = m(R_{t + 1}(F_0), i)$ and recall once again that $L_{t + 1}(G) \maj_{t + 1} R_{t + 1}(F_0)$. 
Thus 
$X$ is obtained from $R_{t + 1}(F_0)$ by removing the $i$th element of $R_{t + 1}(F_0)$ and adding the $i$th element of $L_{t + 1}(G)$. Hence by Proposition \ref{move2right} we have that $L_{t + 1}(G) \maj_{t + 1} X$. On the other hand obviously $X\maj_{t + 1} R_{t + 1}(F_1)$, which means that \eqref{in} is proved. 
\end{proof}

\subsection{Probabilistic lemma }

To upperbound $|\F|\cdot|\G|$, where $\F, \G\subseteq\binom{[n]}{a}$ are $t$-far and $\F$ is an ideal of the order $\maj_a$, we use an approach suggested
in~\cite{frankl1987forbidden}. We introduce a probabilistic measure
$\mu_p$ on the set $2^{[n]}$ such that the probability of a subset $X\in 2^{[n]}$ is equal to
$p^{|X|}(1-p)^{n-|X|}$. It is easy to see that this measure is a
product of Bernoulli measures: each point $x$ belongs to a random set
$X$ with probability $p$ and points are included in the set
independently.

\begin{lemma}
\label{prob_bound}
Let $\mathcal{F}, \mathcal{G} \subseteq \binom{[n]}{a}$ be such that $L_{t + 1}(G)\not\maj_{t + 1} R_{t + 1}(F)$ for all $F\in\mathcal{F}, G\in\mathcal{G}$.  Define $\mathbf{X}$ and $\mathbf{Y}$ to be two independent random variables, both distributed according to $\mu_{a/n}$. Then 
$$\Pr[\mathbf{X}\in\mathcal{F}, \mathbf{Y}\in\mathcal{G}] \le 4n \cdot \exp\left(- (a - t - 1)^2/(20a)\right).$$
\end{lemma}
We will use the following form of the Chernoff bound:
\begin{proposition}[\cite{hoeffding1963probability}, Theorem 1]
\label{additive_chernoff}
Let $Z_1, \ldots, Z_l$ be $l$ independent Bernoulli random variables. Assume that each $Z_i$ takes value $1$ with probability $p$. Then for all $\varepsilon \ge 0$:
\begin{align*}
\Pr\left[\sum\limits_{i = 1}^l Z_i \ge (p + \varepsilon) l\right] &\le \exp\left(-D(p + \varepsilon||p)\cdot  l\right) \\
\Pr\left[\sum\limits_{i = 1}^l Z_i \le (p - \varepsilon) l\right] &\le \exp\left(-D(p - \varepsilon||p) \cdot l\right),
\end{align*}
where $D(x||y)$ is the Kullback -- Leibler divergence:
$$D(x||y) = x \ln\left(\frac{x}{y}\right) + (1 - x) \ln\left(\frac{1 - x}{1 - y}\right).$$
\end{proposition}

We also need the following lower bound on   the Kulback -- Leibler divergence:
\begin{proposition}[\cite{topsoe2000some}]
\label{topsoe}
$D(x||y) \ge \frac{(x - y)^2}{2(x + y)}$.
\end{proposition}
From Propositions \ref{additive_chernoff} and \ref{topsoe} we obtain:
\begin{corollary}
\label{our_chernoff}
Let $Z_1, \ldots, Z_l$ be $l$ independent Bernoulli random variables. Assume that each $Z_i$ takes value $1$ with probability $p$. Then for all $\varepsilon \ge 0$:
$$\Pr\left[\sum\limits_{i = 1}^l Z_i \notin [(p - \varepsilon) l, (p + \varepsilon) l] \right] \le 2\exp\left(-\frac{\varepsilon^2 \cdot l}{4p + 2\varepsilon}\right).$$
\end{corollary}

\begin{proof}[Proof of Lemma~\ref{prob_bound}.]
Denote $s = (a - t - 1)$. 
Let $E$ be the event that for all $r \in\{1, \ldots, n\}$ it holds that 
\begin{align*}
|\mathbf{X}\cap [1, r]| &\in \left[\frac{a}{n}\cdot r - s/2, \frac{a}{n}\cdot r + s/2\right]  \\
\mbox{ and } |\mathbf{Y}\cap [1, r]| &\in \left[\frac{a}{n}\cdot r - s/2, \frac{a}{n}\cdot r + s/2\right].
\end{align*}

Let us show that $\mathbf{X}\in\mathcal{F}, \mathbf{Y}\in\mathcal{G}\implies \lnot E$. Indeed, assume for contradiction that there are $X\in\mathcal{F}$ and $Y\in\mathcal{G}$ such that event $E$ holds 
for $\mathbf{X} = X$, $\mathbf{Y} = Y$. 
Note that 
$L_{t + 1}(Y)\not\maj_{t + 1} R_{t + 1}(X)$. 
Hence,  $m(Y, j) > m(X, a - t - 1 + j) = m(X, s + j)$  for some $j\in\{1, \ldots, t + 1\}$. 
Consider $r = m(X, s + j)$. By definition there are exactly $s + j$ elements of $X$ in $[1, r]$. Since event $E$ holds 
for $\mathbf{X} = X$, $\mathbf{Y} = Y$, 
we get:
\begin{equation}
\label{a_equation}
s + j \le \frac{a}{n} \cdot r + s/2.
\end{equation}
On the other hand, there are at most $j - 1$ elements of $Y$ in $[1, m(X, s + j)] = [1, r]$ (this is because $m(Y, j) > m(X, s + j)$). Hence 
\begin{equation}
\label{a_equation_2}
 \frac{a}{n} \cdot r - s/2 \le j - 1
\end{equation}
(we use once again the fact that $E$ holds for $(X, Y)$).
By adding \eqref{a_equation_2} and \eqref{a_equation} we get $0 \le -1$. Thus an implication $\mathbf{X}\in\mathcal{F}, \mathbf{Y}\in\mathcal{G}\implies \lnot E$ is proved. 

In particular, we get:
$$\Pr[\mathbf{X}\in\mathcal{F}, \mathbf{Y}\in\mathcal{G}] \le \Pr[\lnot E].$$
Hence it is enough to upper bound the probability of  $\lnot E$. If $\lnot E$ holds, then for some $r\in\{1, \ldots, n\}$ we have:
\begin{align*}
|\mathbf{X}\cap [1, r]| &\notin \left[\frac{a}{n}\cdot r - s/2, \frac{a}{n}\cdot r + s/2\right] = \left[\left(\frac{a}{n} - \frac{s}{2r}\right)r, \left(\frac{a}{n} + \frac{s}{2r}\right)r \right]  \\
\mbox{ or } |\mathbf{Y}\cap [1, r]| &\notin \left[\frac{a}{n}\cdot r - s/2, \frac{a}{n}\cdot r + s/2\right] = \left[\left(\frac{a}{n} - \frac{s}{2r}\right)r, \left(\frac{a}{n} + \frac{s}{2r}\right)r \right].
\end{align*}
By Corollary \ref{our_chernoff}  both of these events have probability at most 
\begin{align*}
2\exp\left(- \frac{\left(\frac{s}{2r}\right)^2 \cdot r}{4\cdot \frac{a}{n} + 2\cdot \frac{s}{2r}}\right) &= 2\exp\left(- \frac{s^2}{16\cdot \frac{ar}{n} + 4s}\right) \\
&\le 2\exp\left(- \frac{s^2}{16\cdot \frac{an}{n} + 4s}\right) \\
&=2\exp\left(- \frac{s^2}{16a + 4s}\right) \le 2 \exp\left(- \frac{s^2}{20a}\right)
\end{align*}
Hence the probability of the union of these two events is at most twice as large as the last expression. Then by summing over all $a\in\{1, \ldots, n\}$ we get the required bound.
\end{proof}

\subsection{Tying up loose ends --- proof of Theorem \ref{low-intersections}}

Assume that $\F \subseteq \binom{[n]}{a}$ and $\G \subseteq \binom{[n]}{a}$ are $t$-far. By Lemma \ref{reduction_to_compressed_families} there are $\F^\prime, \G^\prime\subseteq \binom{[n]}{m}$ satisfying the following three conditions:
\begin{itemize}
\item $\F^\prime$ and $\G^\prime$ are $t$-far;
\item $|\F^\prime| = |\F|$ and $|\G^\prime| = |G|$;
\item $\F^\prime$ is left-compressed.
\end{itemize}

By Proposition \ref{ideals-filters} we have that $\F^\prime$ is an ideal of $\maj_a$. Then by Lemma \ref{FI-noncomparable} we get that $L_{t + 1}(G) \not\maj_{t + 1} R_{t + 1}(F)$ for all $F\in\F^\prime$ and $G\in \G^\prime$. Hence by Lemma \ref{prob_bound} we have 
\begin{equation}
\label{prob_2}
\Pr[\mathbf{X} \in \F^\prime, \mathbf{Y}\in \G^\prime] \le 4n \exp\left(-(a - t - 1)^2/(20 a)\right),
\end{equation}
where $\mathbf{X}$ and $\mathbf{Y}$ are two independent random variables distributed according to $\mu_{a/n}$. The left hand side of \eqref{prob_2} equals
$$ |\F^\prime| \cdot |\G^\prime| \cdot \left[\left(\frac{a}{n}\right)^a \cdot \left(1 - \frac{a}{n}\right)^{n - a}\right]^2.$$
Finally, from the following lower bound on $\binom{n}{a}$ 
(see \cite[Lemma 2.4.2]{cohen1997covering})
$$\binom{n}{a} \ge \sqrt{\frac{1}{8n \cdot \frac{a}{n} \cdot \frac{n - a}{n}}} \cdot \left(\frac{n}{a}\right)^a \cdot \left(\frac{n}{n - a}\right)^{n - a},$$
we get:
\begin{align*}
|\F| \cdot |\G| &= |\F^\prime| \cdot |\G^\prime| \\
&\le \left[ \left(\frac{n}{a}\right)^a \cdot \left(\frac{n}{n - a}\right)^{n - a} \right]^2 \cdot 4n \exp\left(-(a - t - 1)^2/(20 a)\right)\\
&\le \left[8n \cdot \frac{a}{n} \cdot \frac{n - a}{n} \cdot \binom{n}{a}^2 \right] \cdot 4n \exp\left(-(a - t - 1)^2/(20 a)\right)\\
&= 32 a (n - a) \cdot \exp\left(-(a - t - 1)^2/(20 a)\right)
\cdot \binom{n}{a}^2.
\end{align*}

\section{Communication lower bound}
\label{sec_com}
Our proof of Theorem~\ref{communication_lower_bound} relies on Theorem \ref{low-intersections}.
Since we are dealing with $k$-party setting, we need the following $k$-dimensional generalization of Theorem \ref{low-intersections}. Fortunately, this generalization  can be obtained 
via a~very simple induction argument.

\begin{lemma}
\label{main_lemma}
For all $n, a, t, k\in\mathbb{N}$ satisfying $t < a < n$ the following holds.  Assume that $\mathcal{F}_1, \mathcal{F}_2, \ldots, \mathcal{F}_k\subseteq \binom{[n]}{a}$ are such that 
$$|\mathcal{F}_i| \ge 2^{k - 2} \cdot \sqrt{32a(n - a)} \cdot \exp\left(-\frac{(a - t - 1)^2}{40 a}\right)\cdot  \binom{n}{a} + 2^{k - 2}$$ 
for all $i\in\{1, 2, \ldots, k\}$. 
Then  there are $F_1\in\mathcal{F}_1, F_2\in\mathcal{F}_2, \ldots, F_k\in\mathcal{F}_k$ such that $|F_1\cap F_i| \ge t + 1$ for all $i\in\{2, \ldots, k\}$.  
\end{lemma}
\begin{proof}

For $t < a < n$ 
let $A_{a,t}^{k,n}$ be the minimal positive integer $N$ such that
for all $\mathcal{F}_1, \ldots, \mathcal{F}_k\subseteq\binom{[n]}{a}$ the following holds. If $|\mathcal{F}_i| \ge N$ for all $i\in\{1, \ldots, k\}$, then there are $F_1\in\mathcal{F}_1, F_2\in\mathcal{F}_2, \ldots, F_k\in\mathcal{F}_k$ such that $|F_1\cap F_i| \ge t + 1$ for all $i\in\{2, \ldots, k\}$.

Let us verify that $A_{a,t}^{k,n}$ 
are non-decreasing in $k$, i.e.:
\begin{equation}
\label{A_non-decreasing}
A_{a,t}^{k,n} \le A_{a,t}^{k + 1,n}
\end{equation}
for all $k \ge 2$ and $t < a < n$.
Indeed, take $\F_1, \F_2, \ldots, \F_k \subseteq\binom{[n]}{a}$  such that $|\F_i| \ge  A_{a,t}^{k + 1,n}$ for all $i\in\{1, \ldots, k\}$. It is clear that  $A_{a,t}^{k + 1, n} \le \binom{n}{a}$. So, from  definition of $A_{a,t}^{k + 1,n}$ applied to the families $\F_1, \F_2, \ldots, \F_{k +1}$, where $\F_{k + 1} = \binom{[n]}{a}$, we conclude that there are  $F_1\in\mathcal{F}_1, F_2\in\mathcal{F}_2, \ldots, F_{k + 1}\in\mathcal{F}_{k + 1}$ satisfying $|F_1\cap F_i| \ge t + 1$ for all $i\in\{2, \ldots, k + 1\}$.  

Theorem \ref{low-intersections} implies that:
$$A^{2,n}_{a,t} \le \left\lfloor \sqrt{32a(n - a)} \cdot \exp\left(-\frac{(a - t - 1)^2}{40 a} \right) \cdot \binom{n}{a} \right \rfloor + 1.$$
Indeed, assume that $\F_1, \F_2 \subseteq \binom{[n]}{a}$ are such that:
$$|\F_1|, |\F_2| \ge \left\lfloor\sqrt{32a(n - a)} \cdot \exp\left(-\frac{(a - t - 1)^2}{40 a} \right) \cdot \binom{n}{a} \right \rfloor + 1.$$
Then $|\F_1| \cdot |\F_2|$ is strictly larger than $32a(n - a) \cdot \exp\left(-\frac{(a - t - 1)^2}{20 a}\right) \cdot \binom{n}{a}^2$. By Theorem \ref{low-intersections} this means that there are $F_1\in\F_1, F_2\in\F_2$ such that $|F_1\cap F_2| \ge t + 1$.

To show the lemma it is enough to demonstrate 
that
$$A_{a,t}^{k+1,n} \le 2 \cdot A_{a,t}^{k,n},$$
for all $k\ge 2$ and $t < a < n$.
To do so, fix $k + 1$ families 
$\mathcal{F}_1, \ldots, \mathcal{F}_{k + 1}\subseteq\binom{[n]}{a}$. Assume that $|\mathcal{F}_i| \ge 2 \cdot A_{a,t}^{k,n}$ for all $i\in\{1, \ldots, k + 1\}$. Our goal is to show that there are 
$F_1\in \mathcal{F}_1, \ldots, F_{k + 1}\in \mathcal{F}_{k + 1}$ 
satisfying
$$|F_1\cap F_i| \ge t + 1  \mbox{ for all } i\in\{2, \ldots, k + 1\}.$$

Denote $N = A_{a,t}^{k,n}$. We claim that there are $N$ distinct $G_1, \ldots, G_N \in \mathcal{F}_1$ such that for every $j\in\{1, 2, \ldots, N\}$ there are $F_2^j\in\mathcal{F}_2, \ldots, \mathcal{F}_k^j \in\mathcal{F}_k$ satisfying $|G_j\cap F_i^j| \ge t + 1$ for all $i\in\{2, \ldots, k\}$.  

We construct such $G_1, \ldots, G_N$ one by one. Assume that $G_1, \ldots, G_j$ for some $j\in\{0, \ldots, N - 1\}$ are already constructed.  Notice that:
$$\left|\mathcal{F}_1\setminus\{G_1, \ldots, G_j\}\right| \ge 2 \cdot A_{a,t}^{k,n} - j \ge  2 \cdot A_{a,t}^{k,n} - N =  A_{a,t}^{k,n},$$
$$|\mathcal{F}_i| \ge 2\cdot A_{a,t}^{k,n} \ge  A_{a,t}^{k,n}, \qquad i = 2, \ldots, k.$$
This means  by definition of $A_{a,t}^{k,n}$ that there are $G\in \mathcal{F}_1\setminus\{G_1, \ldots, G_j\}$, $H_2\in\mathcal{F}_2, \ldots, H_k\in\mathcal{F}_k$ satisfying:
$$|G\cap H_i| \ge t + 1  \mbox{ for all } i \in\{2, \ldots, k\}.$$ 
Then we set $G_{j + 1} = G, F_2^{j + 1} = H_2, \ldots, F_k^{j + 1} = H_k$. Note that $G_{j + 1}$ is distinct from $G_1, \ldots, G_j$ because $G\notin \{G_1, \ldots, G_j\}$.

Finally, consider two families $\{G_1, \ldots, G_N\}$ and $\mathcal{F}_{k + 1}$. These two families are both of size at least $N = A_{a,t}^{k,n} \ge A_{a,t}^{2,n}$ (the last inequality here is by \eqref{A_non-decreasing}). Hence there are $j\in\{1, \ldots, N\}$ and $H_{k + 1} \in \mathcal{F}_{k + 1}$ such that $|G_j\cap H_{k + 1}| \ge t + 1$.  To finish the proof set $F_1 = G_j, F_2 = F_2^j, \ldots, F_k = F_k^j$ and $F_{k + 1} = H_{k + 1}$. 
\end{proof}

We are now ready to prove Theorem \ref{communication_lower_bound}.

\begin{proof}[Proof of theorem  \ref{communication_lower_bound}]

Set $a = \lfloor n/k\rfloor$ and $t = \lfloor(1 - \gamma/4)a\rfloor$. Note that
\begin{equation}
\label{a_t_relation}
a - t \ge \frac{\gamma a}{4} \ge \frac{\gamma (n/k - 1)}{4} = \frac{n}{4 \cdot \frac{k}{\gamma}} - \frac{\gamma}{4} \ge 25\sqrt{n} - \frac{1}{4}
\end{equation}
Here we use the assumption that $\frac{k}{\gamma} \le \frac{\sqrt{n}}{100}$. 
In particular, \eqref{a_t_relation} implies $t < a$ 
for all large enough $n$.

 Define
$$ \mathcal{D} = \left\{ (X_1, \ldots, X_k) \in \binom{[n]}{a}^k : X_1, X_2, \ldots, X_k \mbox{ are disjoint}\right\},$$
$$ \mathcal{I} = \left\{ (F_1, \ldots, F_k) \in \binom{[n]}{a}^k : |F_i\triangle F_{i^\prime}| \le \gamma a \mbox{ for all } i, i^\prime\in\{1, \ldots, k\}\right\}.$$

Observe that:
\begin{equation}
\label{size_of_D}
|\mathcal{D}| = \binom{n}{a} \cdot \binom{n - a}{a} \cdot \ldots \cdot \binom{n - (k - 1) \cdot a}{a} > 0.
\end{equation}

Assume that there is a $c$-bit non-deterministic communication protocol for 
$\mathrm{DISJ}^\prime_{k,\gamma}(n)$. Hence there is a cover of $\mathcal{D}$ by at most $2^c$ boxes  which are  disjoint with $\mathcal{I}$. Among these boxes there is  one which contains at least $|\mathcal{D}|/2^c$ elements of $\mathcal{D}$. Let this box be $\mathcal{F}_1\times\ldots \times\mathcal{F}_k$ for some $\mathcal{F}_1, \ldots, \mathcal{F}_k\subseteq \binom{[n]}{a}$.

Let us show that for some $i\in\{1, \ldots, k\}$ it holds that
\begin{equation}
\label{i_we_need}
|\F_i| < 2^{k - 2} \cdot \sqrt{32a(n - a)} \cdot \exp\left(-\frac{(a - t - 1)^2}{40 a}\right)\cdot  \binom{n}{a} + 2^{k - 2}.
\end{equation}
Indeed, assume that it is not true. Then, since $t < a< n$, we can apply Lemma \ref{main_lemma} to find  $F_1\in\F_1, F_2\in\F_2, \ldots, F_k\in\F_k$ such that $|F_1\cap F_i| \ge t + 1$ for all $i\in\{2, \ldots, k\}$. Note also that $|F_1\cap F_1| = |F_1| = a \ge t+1$.
From that for every $i,i^\prime\in\{1, \ldots, k\}$ we obtain:
\begin{align*}
|F_i\triangle F_{i^\prime}| &\le |F_i\triangle F_1| + |F_{i^\prime} \triangle F_1|\\
 &= |F_i| + |F_1| - 2|F_i\cap F_1| + |F_{i^\prime}| + |F_1| - 2|F_{i^\prime}\cap F_1| \\
&\le 4 \cdot a - 4 \cdot (t + 1) \le  \gamma a.
\end{align*}
This means that $\F_1\times \F_2\times\ldots \times\F_k$ intersects $\mathcal{I}$, contradiction.

Take any $i\in\{1, 2, \ldots, k\}$ satisfying \eqref{i_we_need}. Recall that by definition there are at least $|\mathcal{D}|/2^c$ elements of $\mathcal{D}$ in  $\F_1\times \F_2\times\ldots \times\F_k$. On the other hand,  notice that for any fixed $X\in\binom{[n]}{a}$ there are exactly  
$$\binom{n - a}{a} \cdot \ldots \cdot \binom{n - (k - 1) \cdot a}{a}. $$
elements of $\mathcal{D}$ with the $i$th coordinate  
equals to $X$. 
Hence there are at most
$$ |\F_i| \cdot \binom{n - a}{a} \cdot \ldots \cdot \binom{n - (k - 1) \cdot a}{a}$$
elements of $\mathcal{D}$ in $\F_1\times \F_2\times\ldots \times\F_k$. By combining these two bounds we obtain:
$$|\mathcal{D}|/2^c \le  |\F_i| \cdot \binom{n - a}{a} \cdot \ldots \cdot \binom{n - (k - 1) \cdot a}{a}.$$
By \eqref{size_of_D} this transforms to 
$$2^c \ge \frac{\binom{n}{a}}{|\F_i|}.$$
Recall that the size of $\F_i$ satisfies \eqref{i_we_need}. This gives us the following:
\begin{align*}
2^c &\ge \frac{\binom{n}{a}}{2^{k - 2} \cdot \sqrt{32a(n - a)} \cdot \exp\left(-\frac{(a - t - 1)^2}{40 a}\right)\cdot  \binom{n}{a} + 2^{k - 2}} \\
&\ge \frac{1}{2}\min\left\{\frac{\exp\left(\frac{(a - t - 1)^2}{40 a}\right)}{2^{k - 2}\sqrt{32a(n - a)}}, \frac{\binom{n}{a}}{2^{k - 2}}\right\} \\
&\ge   \frac{1}{2}\min\left\{\frac{\exp\left(\frac{(a - t - 1)^2}{40 a}\right)}{2^{k - 2}\sqrt{32} \cdot n}, \frac{\binom{n}{a}}{2^{k - 2}}\right\}.
\end{align*}
After taking $\log_2$ of the last inequality (bearing in mind that $\log_2(e) > 1$) we obtain that $c$ for all large enough $n$ satisfies the following:
$$c \ge \min\left\{ \frac{(a - t - 1)^2}{40 a} - k  - 1.5\log_2(n) , \log_2\left(\binom{n}{a}\right) - k\right\}$$
(here we subtract $0.5\log_2(n)$ from the first argument of $\min$ to compensate negative constant terms).
It remains to demonstrate that both expressions in the minimum above are at least $\frac{\gamma^2 n}{10^4 \cdot k} - 2\log_2(n)$ for all large enough $n$:
\begin{align}
\label{first_min}
\frac{(a - t - 1)^2}{40 a} - k  - 1.5\log_2(n) \ge \frac{\gamma^2 n}{10^4 \cdot k} - 2\log_2(n),\\
\label{second_min}
\log_2\left(\binom{n}{a}\right) - k  \ge \frac{\gamma^2 n}{10^4 \cdot k} - 2\log_2(n).
\end{align}

Let us start with \eqref{first_min}. At first, note that
$$a - t - 1 \ge \frac{\gamma a}{4} - 1 \ge \frac{\gamma a}{8},$$
where the last inequality is because $\gamma a$ is large enough:
$$\gamma a \ge \gamma (n/k - 1)  \ge 100\sqrt{n} - 1.$$
(in the second inequality of the last line  we use the assumption that $\frac{k}{\gamma} \le \frac{\sqrt{n}}{100}$). In particular,
 $a - t - 1$ is positive. Hence
\begin{align*}
\frac{(a - t - 1)^2}{40 a} - k  - 1.5\log_2(n) &\ge \frac{\gamma^2 a}{2560} - k - 1.5\log_2(n) \\
&\ge \frac{\gamma^2 n}{2560 \cdot  k} - k - 2\log_2(n)
\end{align*}
(here once again we subtract $0.5\log_2(n)$ to compensate a negative constant term which is due to rounding of $a = \lfloor n/k\rfloor$).
To prove \eqref{first_min} it remains to notice that $k \le \frac{\gamma^2 n}{10^4 \cdot k}$ because $\frac{k}{\gamma} \le \frac{\sqrt{n}}{100}$. 

To show \eqref{second_min} we will actually show that the left hand side of \eqref{second_min} is at least the left hand side of \eqref{first_min}. Indeed, 
$$\log_2\left(\binom{n}{a}\right) \ge a \cdot \log_2\left(\frac{n}{a}\right) \ge a.$$
The last inequality is because $k\ge 2$ and hence $a = \lfloor n/k\rfloor \le n/2$. But recall that $a - t - 1$ is positive, which implies that $a$ is  at least $\frac{(a - t - 1)^2}{40 a}$.
\end{proof}

\bigskip

\textbf{Acknowledgment.} The article was prepared within the framework of the HSE University Basic Research Program and funded by the Russian Academic Excellence Project '5-100'. Mikhail Vyalyi is partially supported  by RFBR grant 17--01-00300 and by 
the state assignment topic no. 0063-2016-0003.

Authors are sincerely grateful to anonymous reviewers for valuable comments.

\appendix

\section{Reduction to finite time}
\label{sec_reduction}
\begin{proposition}
Assume that a deterministic finite automaton $\A$ with $q$ states separates $\mathrm{EvenCycles}_{n,d}$ from $\mathrm{OddCycles}_{n,d}$. Then $\A$ separates $\mathrm{EvenCycles}_{n,d}$ from $\mathrm{OddCycles}_{n,d}$ in time $q n$. 
\end{proposition}
\begin{proof}
Let $Q$ be the set of states of $\A$ and let $q_{start}$ be the initial state of $\A$.
Without loss of generality we may assume that $q_{accept}$ is an absorbing state of $\A$, i.e.,
$$\delta_\A(q_{accept}, a) = q_{accept}$$
for all 
$a\in [n] \times \{1, 2,\dots,d\}$. 
Thus it is enough to show that for every $w\in\mathrm{EvenCycles}_{n,d}$ there exists $i \in \{1, 2, \ldots, qn\}$ such that
$$\delta_\A(q_{start}, w_1 \ldots w_i) =q_{accept}.$$
Assume that for some $w = (v_1, l_1) (v_2, l_2) (v_3, l_3) \ldots \in\mathrm{EvenCycles}_{n,d}$ this is false. Let $G$ be an even game graph with at most $n$ nodes which has an infinite path corresponding to $w$.
 Define a mapping $\phi\colon[qn + 1] \to [n] \times Q$ as follows:
$$\phi(i) = (v_{i}, \delta_\A(q_{start}, w_1\ldots w_{i - 1})).$$
By the pigeonhole principle 
there are $i, j\in[qn + 1]$, $i < j$ such that
$$\phi(i) = \phi(j).$$
I.e., $v_i = v_j$ and $\delta_A(q_{start}, w_1\ldots w_{i - 1}) = \delta_A(q_{start}, w_1\ldots w_{j - 1})$. Consider the following infinite path $w^\prime$ of $G$. This path starts at $v_1$ and goes to $v_i$ by edges encoded in $w_1\ldots w_{i - 1}$. Then it stays forever on a cycle starting at $v_i = v_j$ and formed by edges encoded in $w_i\ldots w_{j - 1}$. It is easy to see that the only states $\A$ reaches on $w^\prime$ are:
$$q_{start}, \delta_\A(q_{start}, w_1), \ldots, \delta_\A(q_{start}, w_1 w_2\ldots w_{j - 1}).$$
By our assumption $\delta_\A(q_{start}, w_1), \ldots, \delta_\A(q_{start}, w_1 w_2\ldots w_{j - 1})$ are all different from $q_{accept}$ (and $q_{start}$ is obviously too, because otherwise $\A$ reaches $q_{accept}$ on every word).  On the other hand $w^\prime$ is an infinite path of an even game graph with at most $n$ nodes, contradiction.
\end{proof}

\end{document}